\documentclass{IEEEtran}

\usepackage{cite}
\usepackage{amsmath,amssymb,amsfonts}
\usepackage{amsthm}

\usepackage{algorithmic}
\usepackage{graphicx}
\usepackage{textcomp}

\usepackage{booktabs}
\usepackage[table,xcdraw]{xcolor}

\def\BibTeX{{\rm B\kern-.05em{\sc i\kern-.025em b}\kern-.08em
    T\kern-.1667em\lower.7ex\hbox{E}\kern-.125emX}}

\DeclareMathOperator{\sgn}{sgn}

\usepackage{url}            
\usepackage{xspace}
\usepackage{enumitem}
\newcommand{\etal}{\textit{et~al}.\@\xspace} 

\newcommand{\exposed}{\textit{exposed}\@\xspace} 
\newcommand{\susceptible}{\textit{susceptible}\@\xspace} 
\newcommand{\infected}{\textit{infectious}\@\xspace} 
\newcommand{\removed}{\textit{removed}\@\xspace} 

\newtheorem{theorem}{Theorem}
\newtheorem{lemma}{Lemma}

\newtheorem{definition}{Definition}

\begin{document}

\title{Models for digitally contact-traced epidemics}
\author{Chiara Boldrini, Andrea Passarella, and Marco Conti
\thanks{All authors are with CNR-IIT, Via G. Moruzzi 1, 56124, Pisa (e-mail: first.last@iit.cnr.it).}
\thanks{This work was partially funded by the SoBigData++, HumaneAI-Net, and SAI projects. The SoBigData++ and HumaneAI-Net projects have received funding from the European Union's Horizon 2020 research and innovation programme under grant agreements No 871042 and No 952026, respectively. The SAI project is supported by the CHIST-ERA grant CHIST-ERA-19-XAI-010, funded by by MUR (grant No. not yet available), FWF (grant No. I 5205), EPSRC (grant No. EP/V055712/1), NCN (grant No. 2020/02/Y/ST6/00064), ETAg (grant No. SLTAT21096), BNSF (grant No. KP-06-DOO2/5).}
\thanks{This work has been published in  IEEE Access, vol. 10, pp. 106180-106190, 2022, doi: 10.1109/ACCESS.2022.3211425 under a Creative Commons License.}
}



\maketitle

\begin{abstract}
Contacts between people are the main drivers of contagious respiratory infections. For this reason, limiting and tracking contacts is a key strategy for controlling the COVID-19 epidemic. Digital contact tracing has been proposed as an automated solution to scale up traditional contact tracing. However, the required penetration of contact tracing apps within a population to achieve a desired target in controlling the epidemic is currently under discussion within the research community. In order to understand the effects of digital contact tracing, several mathematical models have been studied. In this article, we propose a novel compartmental SEIR model with which it is possible, differently from the models in the related literature, to derive closed-form conditions regarding the control of the epidemic. These conditions are a function of the penetration of contact tracing applications and testing efficiency. Closed-form conditions are crucial for the understandability of models, and thus for decision makers (including digital contact tracing designers) to correctly assess the dependencies within the epidemic. Feeding COVID-19 data to our model, we find that digital contact tracing alone can rarely tame the epidemic: for unrestrained COVID-19, this would require a testing turnaround of around 1 day and app uptake above 80\% of the population, which are very difficult to achieve in practice. However, digital contact tracing can still be effective if complemented with other mitigation strategies, such as social distancing and mask-wearing. 
\end{abstract}




\section{Introduction}
\label{sec:intro}

\IEEEPARstart{S}{ince} April 2020, the WHO has been recommending two main and complementary strategies to curb the COVID-19 epidemic: social distancing on the one end, \emph{test, trace, treat} (the famous 3 T's) on the other. As for any respiratory viral infection, the sooner we are able to ``remove'' contagious people from interacting with others, the sooner the epidemic will be restrained. Indeed, if, on average, each infected person infects less than one susceptible person, rather than, for example, two or three, the epidemic will die naturally~\cite{Fraser2004}. Of course, for this to be effective, all three T's must be carried out swiftly. Contact tracing without testing is impossible: first, you have to know that a person is potentially contagious before being able to track down their contacts. Similarly, tracing must be completed as quickly and as thoroughly as possible: the longer it takes to identify past contacts, the longer the time a potentially contagious person spends unknowingly infecting other people. Then, contagious people must be immediately isolated, and treated if necessary. 

Contact tracing can be performed manually or digitally. Manual contact tracing involves reconstructing the history of past contacts with the infected person in the days before being detected as contagious. This is typically done through interviews with the infected person. Manual contact tracing suffers from two main problems: i) it is labor-intensive, hence it struggles to keep up when the number of daily new cases is high, and ii) the contagious person might not be able to recall precisely their past contacts (simply because they forget some or because some chance contacts with strangers are not noticed in the first place). 
Digital contact tracing, performed by means of smartphone apps that -- typically via Bluetooth -- automatically detect and register contacts, have the potential to overcome the two limitations of manual contact tracing described above. The research community has already identified convincing solutions that provide reasonable trade-offs between privacy and tracing accuracy~\cite{Cho2020,Raskar2020,Nanni2020,Oliver2020}. Specifically, decentralized Bluetooth-based contact tracing has emerged as the solution of choice and privacy-preserving apps based on this approach have been deployed\footnote{For an extensive list: \url{https://en.wikipedia.org/wiki/COVID-19_apps}.} in many countries~\cite{Ahmed2020,Troncoso2022}. Recent research proposals involve leveraging blockchains and IoT for a more efficient and privacy-preserving tracking~\cite{Garg2020}.
The vast majority of deployed apps leverage the Exposure Notification protocol, jointly rolled out by Apple and Google in Spring~2020. However, digital contact tracing comes with its own problems. The main one is that, for it to be effective, a significant percentage of the population must have the app installed~\cite{Ferretti2020}. 
Bumping up this percentage may not be as easy as it seems~\cite{Li2020c}. For example, people with old smartphones (typically not supporting Bluetooth Low Energy or for which an updated operating system is not available) cannot enjoy the tracking functionality. Fear of government intrusion into privacy turns off other potential participants.  

Due to the limitations discussed above, the percentage of people with an installed and fully-functioning contact tracing app will be far from 100\%. Thus, key questions are, among others: how large should this percentage be for digital contact tracing to be effective in containing the epidemic? How does this percentage depend on the contact patterns between people? How does it depend on the implementation of other mitigation measures (such as social distancing and mask-wearing)?
Digital contact tracing has yet to be properly evaluated as a public health measure through a large-scale assessment~\cite{Colizza2021}.
Thus, the answers to the above questions must then necessarily come from mathematical models and simulations. 
The network of people with the contact tracing app installed is just another instance of a mobile social network~\cite{CONTI20122}: people interact with each other socially and these interactions are mirrored in the anonymous data collected from the contact tracking app. Thus, by taking advantage of the properties of this mobile social network, we will investigate the above questions. 

Ferretti~\etal~\cite{Ferretti2020} adapted a model introduced by Fraser~\etal~\cite{Fraser2004} (and based on the popular Von Foerster equation) in order to tackle the same research problem. However, the model in~\cite{Ferretti2020} can only be solved numerically, and hence it is unable to yield a closed-form condition under which the epidemic can be controlled based on the characteristics of the digital contact tracing in place. In this article, to complement the model in~\cite{Ferretti2020}, our goal is to propose a deterministic compartmental model for digital contact tracing that provides closed-form conditions for the control of an epidemic. \emph{Closed-form control conditions are crucial for the understandability of models and are instrumental for decision makers and computer scientists working on digital contact tracing.}

The contributions of this paper are the following:
\begin{itemize}
    \item To complement the model in~\cite{Ferretti2020}, we propose a deterministic compartmental model for digital contact tracing. The advantage of this modeling approach is that closed-form solutions can be obtained, and hence analytical conditions on the control of the epidemic can be derived. Vice versa, the Von Foerster equation on which~\cite{Ferretti2020} is based, is more accurate than simple compartmental models, but can only be solved numerically. 
    \item Alongside the compartmental model, we also introduce a standalone model that captures how the testing delay affects the efficacy of the detection of infected people, depending on the duration of the latent window (the period during which an infected person is not yet contagious) and the contagious window (during which the infected person is contagious but has yet to develop symptoms). This model is general, and can be solved in closed form for some common distributions describing these time intervals.
    \item We apply the model to realistic COVID-19 epidemic scenarios, showing that, even with high penetration of digital contact tracing, the control of the epidemic is extremely difficult without additional mitigation strategies.
\end{itemize}


The rest of the paper is organized as follows. In Sections~\ref{sec:covid_primer}-\ref{sec:models}, we overview the main results in the related literature regarding COVID-19 modeling and we summarise the properties of the disease itself that are important from the modeling standpoint. Our deterministic compartmental model is presented in Section~\ref{sec:seir}, together with the model on the efficacy of detection of infected persons. The proposed model is then applied to a set of realistic epidemiological scenarios in Section~\ref{sec:evaluation}.
Finally, Section~\ref{sec:conclusion} concludes the paper.

\section{A brief overview on COVID-19}
\label{sec:covid_primer}

From the modeling standpoint, a crucial aspect is to understand when infected people become contagious. For any viral disease, the typical timeline is the following. Following contact with a contagious person, an individual may become infected. However, they do not become contagious immediately: there is a \emph{latent period} during which the person is infected but not yet contagious (i.e., the virus is replicating but its quantity is not yet enough to infect another person). Another important stage is the \emph{incubation period}, which goes from the infection time to the time when the person starts developing symptoms. The latent period may be shorter than the incubation period: this means that an infected person becomes contagious before developing symptoms. Clearly, this makes controlling the spread of the disease harder, since the contagious person who has not developed any symptoms is not aware of their contagiousness. Despite being the subject of hot debates in the initial phases of the COVID-19 epidemic, it is now clear that asymptomatic and pre-symptomatic carriers play a major role in the spread of the SARS-CoV-2 virus~\cite{Bohmer2020,Hu2020,Qian2020,Rothe2020,Wang2020,Yu2020,Bai2020,Wei2020,Wolfel2020,Kimball2020,Li2020a,Li2020b,Sutton2020,He2020,Luo2020}.

SARS-CoV-2 is an airborne\footnote{We use the layman's definition of airborne here. Technical use implies only transmission through aerosol.} virus, i.e., it travels through the air. The typical transmission pathway occurs when a contagious person talks, sneezes, or coughs, producing infectious droplets that are inhaled by people in close proximity. Less frequently, these droplets may fall on the surfaces in close proximity and then contribute to transmitting the disease when the contaminated surface is touched by a susceptible person and then this person touches his/her face (eyes, mouth, etc.). The latter transmission pathway is known as environmental transmission. It is not known to play a major role in the COVID-19 epidemic\footnote{https://www.ecdc.europa.eu/en/covid-19/latest-evidence/transmission},
hence we will not consider it in the modeling. A third transmission pathway is that of aerosol~\cite{Morawska2020,chagla2020airborne-transmission,Klompas2020,Bazant2021}:
 when a contagious person talks, sneezes, or coughs they also produce some smaller droplets (known as aerosol) that evaporate faster than they fall on the ground~\cite{Bourouiba2020}.
This means that with aerosol transmission, the dry virus lingers in the air for a considerable time and travels long distances. The bad news is that common face masks (such as surgical and cloth ones) are not well equipped to contain such small droplets. Thus, aerosol transmission is much more challenging than droplet transmission, which can be easily contained by relying on widespread social distancing and lower-grade mask-wearing.  However, model-wise, they can both be captured by appropriately tuning the probability of infection upon contact.

\section{Models of epidemics}
\label{sec:models}

There are two main modeling approaches in the related literature: mathematical models and agent-based models. 
Mathematical models of epidemics typically lay out a system of Ordinary Differential Equations (ODE)/Partial Differential Equations (PDE) that describes how the number of susceptibles, infected, etc., varies over time. 
Sometimes these systems can be solved in closed form and provide very useful trends describing what-if scenarios. 
Otherwise, numerical solutions can be obtained. 
Due to their nature, these models are based on several simplifying initial assumptions to make the mathematical representation of the phenomenon tractable. 
On the opposite side of the modeling spectrum, there are agent-based models. 
Agent-based models are computational models where agents (corresponding to people) interact, in simulation, according to some predefined rules, which can be arbitrarily complex~\cite{Hunter2018,Venkatramanan2018,Hackl2019,Parker2011,Bonabeau2002,Grefenstette2013,Liu2015}. 
They are conceptually very similar to the models used in transportation simulation. They recreate synthetic populations in terms of demographics, traffic flows, etc., and then an epidemic is simulated. Recently, machine learning approaches have gained popularity. Tomy~\etal~\cite{Tomy2022} exploit Graph Neural Networks (GNN) to solve the problem of inferring the state of the entire population by observing just a few individuals. GNN have also been used to forecast pandemic evolution~\cite{Fritz2022,Gao2021,Kapoor2020}, for the temporal reconstruction of epidemic spreading~\cite{Cutura2021}, and for identifying patient-zero in an epidemic~\cite{Shah2020}. More traditional approaches have also been used, e.g., in~\cite{Dash2021}, for epidemic forecasting.
Since neither agent-based nor machine learning approaches are the focus of this work, we will not discuss them further.

By far, the most used mathematical model is the classical SIR model and its many variations~\cite{Brauer2012}. In the basic version, people are divided into three compartments (denoted with S, I, R, hence the name of the model). In S, people are susceptible to the disease, i.e., they can become infected upon contact with an infectious person. Infectious people are in compartment I. After a certain time spent in compartment I, infected people recover and move to compartment R. Transitions between compartments are then modeled as follows:

 \begin{align*}
  \vspace{-10pt}
 \frac{dS}{dt} &= - \beta I \\ 
 \frac{dI}{dt} &=  \beta I - \gamma I \\
 \frac{dR}{dt} &= \gamma I.
 \vspace{-10pt}
 \end{align*}

Parameters $\beta$ and $\gamma$ describe the rate at which susceptibles become infected and infected recover. The SIR model can be described by a system of ordinary differential equations that can be solved in a closed form. This representation of an epidemic is referred to as \emph{deterministic}, because the above equations are an approximation, holding for very large populations, of the stochastic version\footnote{Stochastic SIR models have been extensively studied in the complex system community, as they are amenable to capture the effect of network topology. At the beginning of a disease outbreak, individual variability (such as whether a node is a ``hub'' in the contact network) plays a major role in determining whether an epidemic will occur or not, and deterministic models (which treat all nodes as equal) are not an appropriate choice.  
In fact, deterministic models assume a fully-mixed population, meaning that an infected individual has the same probability of infecting any susceptible node in the network. This assumption is needed to write down the ODE system, but it is only considered reasonable once an epidemic has started, and several nodes are already contributing to it. Our model, similarly to other ODE-base models, should be used when a disease has already achieved its epidemic phase, in which accurately capturing the outbreak stage is no longer important while having tools for studying the controllability becomes essential. Please refer to~\cite{Brauer2012}, and reference therein, for a detailed discussion.} of the SIR model~\cite{Brauer2012}.  
The simple SIR model has been extended in several directions, adding the \exposed compartment (where people infected but not yet contagious reside), which we also use in this work, and many more (see \cite{Brauer2012} for a general discussion and~\cite{Giordano2020ModellingItaly.} for an application to COVID-19). It has also been used to study the two-pathogen case~\cite{Newman2005}, when two pathogens insist on the same population (we do not consider this case in our work). 

The deterministic compartmental models discussed above are based on the simple assumption that the time spent in each compartment can be reasonably approximated with an exponential random variable (the Markovian assumption). When this is not the case, other types of models must be considered. An important class of non-Markovian models are those based on the McKendrick-VonFoerster equation, which incorporates a so-called age structure to the model~\cite{Brauer2012}. Originally, this model was designed to capture births and deaths in the dynamics of population growth in cellular biology: offspring are generated at a young age, and death occurs typically at an old age. Hence, keeping track of the population age over time was essential to predict the evolution of the population size. When applied to epidemiology, \emph{age} is seen from the point of view of infection, i.e., it corresponds to the time since the individual became infected. And the \emph{birth rate} at infection age $\tau$ becomes the rate at which  a person infected $\tau$ days ago produces offspring, i.e., newly infected people. Thus, the infection rate is no longer constant over time and depends on the current age profile of the population.

Finally, a related active area of study for COVID-19 is the correct estimation of the parameters that describe the dynamics of infection~\cite{Riou2020,ferguson2020report,Linton2020EpidemiologicalData,Li2020,Dorigatti,Arenas2020}. This is important both for purely mathematical models and for agent-based models, because a correct estimation of the epidemic parameters allows researchers to correctly set up their assumptions and simulations.

\section{Factoring in digital contact tracing}
\label{sec:seir}

The McKendrick-VonFoerster model introduced in Sec.~\ref{sec:models} has been used in~\cite{Ferretti2020}, a seminal paper dedicated to evaluating the efficacy of contact tracing for COVID-19. The model by Ferretti \etal~\cite{Ferretti2020} is based on the one proposed in~\cite{Fraser2004}, with parameter values customized to the COVID-19 setting. 
This model is the de-facto reference for digital contact-tracing effectiveness estimation and the vast majority of forecasts, coming both from  within the scientific communities and news outlets, have been based on its results. 
By its own nature, the McKendrick-VonFoerster model can only be solved numerically. Hence, it is unable to yield a closed-form condition under which the epidemic can be controlled based on the characteristics of the digital contact tracing in place. 
Thus, in this work, we complement the results of Ferretti \etal~\cite{Ferretti2020} showing how a simpler model, whose control condition is solvable in closed form, can be obtained.  The notation we use in the paper is summarized in Table~\ref{tab:notation} at the end of the section.
 

\begin{figure}[t]
\begin{center}
\includegraphics[scale=0.45]{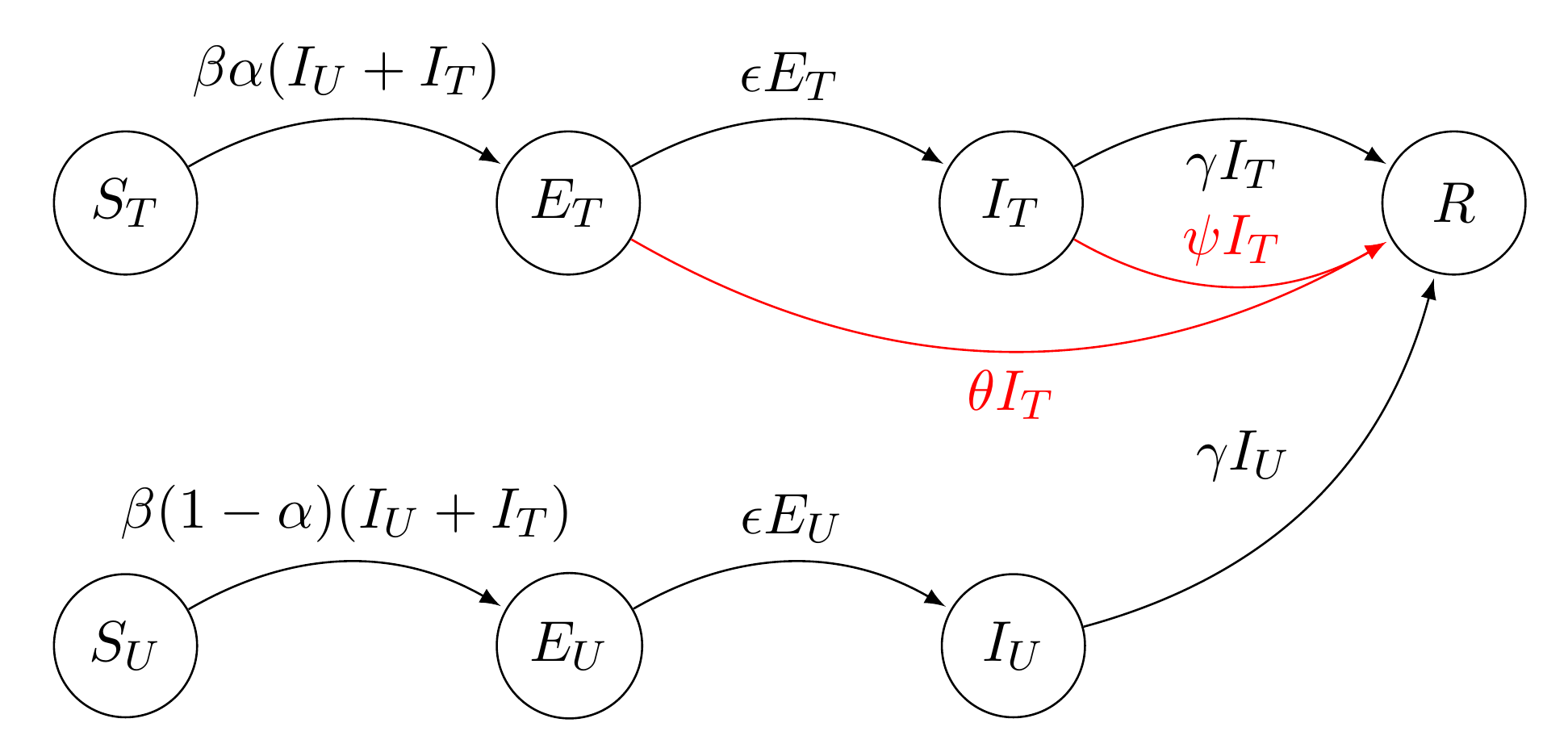}
\caption{How people move across SEIR compartments. In red are the arcs associated with containment measures: quarantine of exposed people plus isolation of those infected and contagious. $\alpha$ denotes the fraction of the digitally tracked population.} \vspace{-20pt}
\label{fig:SEIR_model}
\end{center}
\end{figure}

We start with a simplified version of the model, presented in Figure~\ref{fig:SEIR_model}, for illustrative purposes.
As usual for deterministic compartmental models, we start with a population of constant size $N$, i.e., the sum of people in all states must add up to $N$. When looking at a large population, the short-term variation in its size is small and can be neglected. 
The goal of the model is to capture how infection spreads through the population and to assess whether the spread can be stopped or not by means of digital contact tracing and the resulting quarantine of contacts. 
We assume that the epidemic is faster than the long-term dynamics of births and deaths in the population, so we ignore the latter. 
People can be in one of four states: S (\susceptible), E (\exposed), I (\infected), R (\removed). Since people can be either tracked (with a contact tracing app) or untracked, we duplicate these states to account for this difference (thus, each state is marked with subscript T or U for tracked and untracked, respectively). We do not need to duplicate the \removed state because people in R do not contribute to the epidemic anymore. While we use the common letters S, E, I, R to denote the states, we slightly adjust the default meaning of the states to take into account asymptomatic and presymptomatic transmission. 
In this model, then, \exposed means infected but not yet contagious, while \infected means infected and contagious but with no symptoms (this includes the pre-symptomatic and the asymptomatic phase of the disease). The \removed state includes all infected (whether contagious or not) that have been isolated and/or have recovered. In this simplified model, a person is isolated either because she is infected and has been tracked by the contact tracing app or because she is contagious and has started developing symptoms. We do not include a dedicated state where people are both contagious and symptomatic because it is reasonable to assume that people with symptoms will isolate themselves (and therefore join the \removed state). 
%


The fraction of tracked people is indicated by $\alpha$, where $\alpha$ represents the percentage of the population that subscribed to the considered contact tracing app. Thus, at time $t=0$, we have $\alpha N$ people in state $S_T$ (corresponding to people that are susceptible and tracked) and $(1-\alpha) N$ people in $S_U$ (susceptible but not tracked). From the \susceptible state, people can only move to the \exposed state\footnote{Technically, also some susceptibles can be removed (i.e., asked to isolate) due to, e.g., a faulty detection from the contact tracing app or because the app detected a contact that did not generate an infection. In practice, these notified people will take a test in a matter of days and the test will come up negative, joining again the $S$ compartment. Thus, since SEIR models are based on the assumption $S \sim N$, with $N$ large, this removal can be effectively ignored.}. Recall that ``exposed'', in this case, means infected but not yet contagious. Therefore, the time spent in the \exposed state corresponds, without containment measures in place, to the latent period of the disease. However, there is a crucial difference between untracked and tracked exposed: tracked exposed will be notified by health authorities about their previous contact with a tracked contagious person and they will be isolated, i.e., they will move to the \removed state. 
The same happens for tracked infectious. 
Removed people are isolated, and hence cannot infect anyone. This is the crucial contribution of contact tracing. Below, we summarize how the transitions from each state can be modeled.

\begin{description}
\item[$S_U \rightarrow E_U$] The rate at which people leave the $S_U$ state is given by the effective contact rate $\beta$ (rate at which there is an encounter between an S and an I and this encounter generates an infection) times the number of possible encounters between people in $S_U$ and those in $I$ (we don't care whether the encounter is with a tracked or untracked infected). 
\item[$S_T \rightarrow E_T$] The same holds true for the rate at which tracked susceptibles leave $S_T$, with the appropriate change from U to T of S's subscript with respect to the previous case. 
\item[$E_U \rightarrow I_U$] Untracked exposed (people in $E_U$ state) stay there until they become contagious, and this happens at a rate $\epsilon$.
\item[$E_T \rightarrow I_T$]  Instead, \emph{tracked} exposed will either become contagious and move to state $I_T$ (this happens with the rate $\epsilon$) or be warned of having had contact with an infected and told to self-isolate (this happens with rate~$\theta$, discussed in detail in Sec.~\ref{sec:parameters_from_contact_history}). 
\item[$I_U \rightarrow R$]  Untracked contagious ($I_U$) are isolated when they begin to develop symptoms, and this happens at a rate $\gamma$. 
\item[$I_T \rightarrow R$] Instead, tracked contagious ($I_T$) can be either isolated when they start developing symptoms (similarly to the untracked case) or be informed that they have had contact with an infected and told to self-isolate (this happens with rate~$\psi$, discussed in detail in Sec.~\ref{sec:parameters_from_contact_history}).
\end{description}

The key point in being able to add the effect of contact tracing to the SEIR equations is to adequately model the red transitions in Figure~\ref{fig:SEIR_model}. The rates $\theta$ and $\psi$ capture how effective digital contact tracing is in removing infected people, and factor in testing delay as well as epidemic features (latent period, contagious period, etc.). We will discuss this aspect in Section~\ref{sec:parameters_from_contact_history} and provide a methodology to derive $\theta$ and $\psi$.  For now, let us assume that we can assign proper values to all the parameters of the model. In Theorem~\ref{theo:c1} below, we discuss how to solve the model and how to assess whether the epidemic can be controlled or not depending on the efficacy of contact tracing. The proof of the theorem can be found in Appendix~\ref{app:proof_c1}.

\begin{theorem} \label{theo:c1}
The epidemic described by the SEIR model in Figure~\ref{fig:SEIR_model} can be controlled when the following condition is true:
\begin{equation}\label{eq:c1}
\textrm{C1: } \alpha - \left(1 + \frac{\gamma}{\theta + \psi} \right) \left(1-\frac{\gamma}{\beta}\right) > 0.
\end{equation}
\end{theorem}

\emph{Remark.} The closed-form condition in Theorem~\ref{theo:c1} could not have been obtained with the model used by Ferretti \etal in~\cite{Ferretti2020}, which can only be solved numerically. Closed-form conditions are crucial for the understandability of models, and thus for decision makers (including digital contact tracing designers) to correctly assess the dependencies within the epidemic. Note also that the complexity of this solution is~$O(1)$, hence it does not depend on the size of the population like, e.g., for agent-based models. Thus, \emph{C1 provides an easily interpretable and fast answer to the controllability problem, albeit approximate}.

To illustrate the intuition behind condition C1 in Theorem~\ref{theo:c1}, let us consider two ideal cases separately: i) instantaneous tracing ($\theta +\psi \rightarrow \infty$) and ii) perfect app uptake ($\alpha = 1$). These correspond to the two dimensions of digital contact tracing: how good we are at detecting infections of the tracked people and how many people we are able to track. 
When tracing is instantaneous (corresponding to the first case above), the threshold on $\alpha$ (derived from Equation~\ref{eq:c1}) converges to $1-\frac{\gamma}{\beta}$. For SIR models, the ratio $\frac{\beta}{\gamma}$  corresponds to the basic reproduction number $R_0$~\cite{daley2001epidemic}, hence the threshold on $\alpha$, interestingly, is equivalent to the herd immunity threshold $1-\frac{1}{R_0}$. Note that instantaneous tracing alone is not sufficient to control the epidemic: $\alpha$ must be high enough for tracking to cover a large fraction of the population. A superfast tracking that only follows just a tiny fraction of the population is basically useless.
In the second case (perfect app uptake, i.e. $\alpha = 1$), condition C1 reduces to $\gamma + \theta +\psi > \beta$. Therefore, $\theta +\psi$ must be large enough to compensate for a high $\beta$ (effective contact rate). This means that even under the ideal situation where everyone has the app ($\alpha = 1$), control of the epidemic may not be attainable if the tracing process is slow. The efficiency of contact tracing is captured by $\theta$ and $\psi$ and, in Section~\ref{sec:parameters_from_contact_history}, we discuss how to derive them. %
\subsection{Estimating parameter $\theta$ and $\psi$ from contact history}
\label{sec:parameters_from_contact_history}

Now, we step back and discuss how to model $\theta$ and $\psi$, which are the rates at which exposed and infectious people are removed, respectively. As discussed above, they capture the effectiveness of testing. To derive them, we have to reconstruct the process from contagion (encounter with an infectious person that yields to infection) until removal.

Exposure notifications are triggered by tracked people becoming symptomatic and, therefore, being tested. We know that, since SEIR models assume homogeneity in encounters (which boils down to a single $\beta$ describing the entire contact process, with no distinction between high vs. low social interactions), the contact rate at which the newly symptomatic tracked person met with tracked susceptibles is $\alpha \beta$ (i.e., the baseline rate scaled by the fraction of tracked people). 
This rate must be split across the different states in which the previous contact might currently be in. Specifically, a past contact can be still exposed, already infectious, or removed. We neglect the removal of susceptibles because the population of susceptibles is very large (by assumption, $S \sim N$), hence removing them would not impact the epidemic. Thus, Definition~\ref{def:alertable_contacts} below follows. 

\begin{definition}[Alertable Contacts]\label{def:alertable_contacts}
The \emph{alertable} contacts of a positive person~$i$ can be a) in state $E_T$  (no symptoms, not contagious), b) in state $I_T$ (contagious, no symptoms), c) in state $R$ (symptomatic or recovered, hence already ``removed'' from the epidemic). We denote the probabilities associated with each of these conditions as $p_E$, $p_{I}$, and $p_{R}$, respectively (note that they add up to 1). 
\end{definition}

\noindent
Intuitively, health authorities should strive to increase as much as possible $p_E$, because people in the \exposed state have yet to infect someone. Of course, this might not be possible (e.g., due to testing delays) so the next best thing is to increase $p_I$. Instead, notifying people who are already in the \removed state is completely useless from the point of view of epidemic containment.  

\begin{figure}[t]
\begin{center}
\includegraphics[scale=0.45]{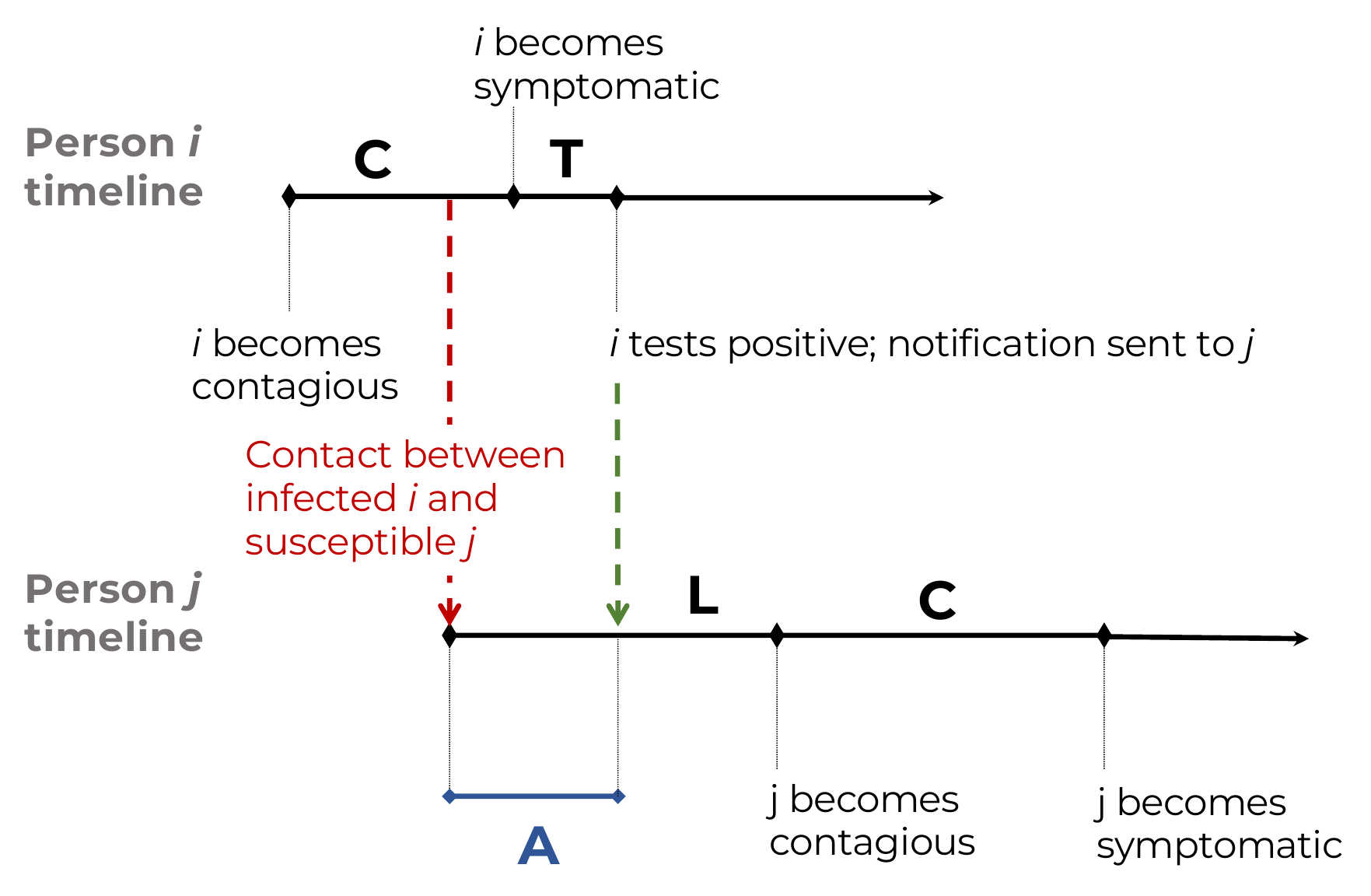}\vspace{-10pt}
\caption{The timeline of the conversion to symptomatic.}\vspace{-20pt}
\label{fig:timeline_to_symptoms}
\end{center}
\end{figure}

As illustrated in Figure~\ref{fig:timeline_to_symptoms}, we can model the conversion to symptomatic of a previous contact considering: the length of the latent period $L$ (which, as discussed in Section~\ref{sec:covid_primer}, goes from the time of infection to the time a person becomes contagious), the length of the infectious but asymptomatic period $C$, and the testing delay $T$ (the time it takes for a test result to be available after the person has developed symptoms). Note that $L$ and $C$ are only determined by the properties of the specific disease. On the contrary, $T$ is totally dependent on the efficacy of the testing system in place, hence it can be shortened by human interventions (e.g., using rapid tests rather than molecular ones or by scaling up testing facilities). The probability distribution of $L$, $C$, and $T$ can be obtained from real data, when available (at the end of the section, we will discuss an example based on a realistic duration of the latent and infectious windows). 
Using their distribution, we can characterize the only missing time interval in Figure~\ref{fig:timeline_to_symptoms}: $A$, which represented the time it takes for the app notification to pop up after contact. In Lemma~\ref{lemma:a} we derive interval $A$'s distribution.

\begin{lemma}\label{lemma:a}
The random variable $A$ describing the time interval between the at-risk contact and the time when the notification from the contact tracing app arrives is distributed as $C' + T$, i.e., as the sum (between random variables) of the residual infectious-but-asymptomatic period~$C'$ and the testing delay~$T$.
\end{lemma}

\begin{proof}
As illustrated in Figure~\ref{fig:timeline_to_symptoms}, $A$ describes the time at which the contact tracing app notification arrives. This time corresponds to the interval between contact with an infectious person and the notification time, hence it includes a \emph{residual} contagious period (which we denote with $C'$) and the testing delay T. Thus, $A$ is distributed as $C' + T$ (hence its PDF is given by the convolution of the PDF of $C'$ and $T$~\cite{dekking2005modern}). Mathematically, $C'$ can be obtained assuming that the contact between the susceptible and the contagious individuals appears during $C$ as a random observer: as a result, $C'$ can be derived as the residual duration~\cite{boldrini2011pareto} of $C$ for the infectious person $i$ (i.e, the time left before the person becomes symptomatic, hence they are discovered to be positive). 
Denoting by $F_X$ the CCDF of a generic random variable $X$, the formula to calculate the residual time $C'$ is the following~\cite{boldrini2011pareto}:
 \begin{equation}
 F_{C'}(t) = \frac{1}{E[C]} \int_t^{\infty} F_C(u) du.
 \label{eq:residual}
 \end{equation}
 Since $C$ can be derived from real epidemic data, $C'$ can also be calculated.
 \end{proof}

Now that $A$ is fully characterized, by deriving its interplay with $L$ and $C$ we obtain $p_E$, $p_I$, and $p_R$ in Lemma~\ref{lemma:probs} below. 

\begin{lemma}\label{lemma:probs}
The probabilities $p_E$, $p_I$, and $p_R$ (associated, respectively, with catching a person in state $E_T$, $I_T$, and $R$) are given by the following:
  \begin{eqnarray*}
 p_E &=&  P(W < 0) \\
 p_I  &=& \int_0^{\infty} P(W=w) P(C > w) \textrm{d}w \\
 p_R &=& P(W-C > 0),
 \end{eqnarray*}
where we denote by W the difference $A-L$.
\end{lemma}

\begin{proof}
From Figure~\ref{fig:timeline_to_symptoms}, we can see that $p_E$ is equivalent to the probability of the notification arriving within the latent period (corresponding to $P(A < L)$). The probability that the notification arrives during the contagious and asymptomatic period ($P(L < A < L + C)$) yields $p_I$. The value of $p_R$ can then be obtained complementing to 1 (or computing $P(A > L + C)$, i.e., the probability that the notification arrives when the individual is already symptomatic). 
Operatively, this results in the thesis.
\end{proof}
 
Not for all distributions the above algebra of random variables yields closed-form solutions, but for some significant ones, it does, at least approximately. This happens, e.g., in the Normal case discussed in the next section (Sec.~\ref{sec:normal_example}). Closed-form solutions can also be obtained with exponential random variables.
 Once the probabilities $p_E$ and $p_{I}$ are derived, it is straightforward to obtain rates $\theta$ and $\psi$.
 
\begin{theorem}\label{theo:psi_theta}
The rates at which exposed and contagious people are removed ($\theta$ and $\psi$, respectively) are given by the following:
 \begin{eqnarray}
 \theta &=& p_E \alpha \beta, \label{eq:theta}\\
 \psi    &=& p_I \alpha \beta, \label{eq:psi}
 \end{eqnarray}
 where $p_E$ and $p_I$ are obtained as in Lemma~\ref{lemma:probs}.
\end{theorem}

\begin{proof}
The thesis simply follows from scaling the overall tracked contact rate $\alpha \beta$ by the probability that the exposed person is notified when still in the \exposed state or in the \emph{infectious} state.
\end{proof}

\begin{table}[t]
\caption{Summary of notation.}
\label{tab:notation}
\begin{tabular}{@{}ll@{}}
\toprule
 \sc{Notation} & \sc{Description} \\ \midrule
     	$N$ & population size \\
	$S$ & susceptible population \\
	$E$ & people that have the disease but not yet symptoms\\
	$I$ & contagious \& symptomatic people \\
	$R$ & removed people (either recovered or isolated) \\
	$( )_T$ & subscript T denotes compartments that are tracked\\
	$()_U$ & subscript U denotes compartments that are \emph{not} tracked \\
        $\beta$ & effective contact rate  \\
        $\epsilon^{-1}$ & time from infected to contagious \\
        $\gamma^{-1}$ & time from contagious to removed (recovery/isolation)\\
        $\theta$ & rate at which tracked exposed are isolated \\
        $\psi$ & rate at which tracked infectious are isolated \\
        $\alpha$ & fraction of population with the app installed and running \\
	$L$ &  latent period \\
	$C$ & infectious-but-asymptomatic period \\
	$T$ &  testing delay \\
	$A$ &  time between at-risk contact and app notification \\
	$W$ &  time to contagious after app notification\\
	$p_E$ & probability of alerting a person in state $E_T$\\
	$p_I$ & probability of alerting a person in state $I_T$\\
	$p_R$ & probability of alerting a person in state $R$\\
\bottomrule
\end{tabular}
\end{table}

\subsubsection{Example with normally distributed characteristic times}
\label{sec:normal_example}

For the sake of example, we can now obtain $p_E$, $p_{I}$, and $p_R$ leveraging the typical average duration of the latent and contagious periods for the original COVID-19 epidemic. From~\cite{Bar-On2020}, we obtain the average duration of the latent period ($\mathbb{E}[L] = \epsilon^{-1}$ = 3 days) and that of the period before an infected becomes contagious ($\mathbb{E}[C] = \gamma^{-1}$ = 2 days). Note that the expectations of $L$ and $C$ correspond to the inverse of $\epsilon$ and $\gamma$ in the SEIR model of Section~\ref{sec:seir}. For example, assume that $L$ and $C$ are normally distributed, each with a standard deviation 0.5 (the occurrence of negative values with this configuration is negligible). We also assume that $T$ is normally distributed, with rate $\mu_T$ and standard deviation~$\sigma_T$. 
%
It is easy to verify that $A = C' + T$ can be approximated as normally distributed as well, specifically $A \sim \mathcal{N}(\mathbb{E}[C'] + \mu_T, \textrm{Var}(C') + \sigma_T^2)$. Since we are dealing with normally distributed variables, it is easy to obtain their difference and sum using the algebra of normally distributed random variables. 

Leveraging the formulas we have obtained, we can now better understand the impact of testing delays on the ability to intercept infected people in each stage using contact tracing. In the following, we focus on a tagged pair of people (one tracked infectious $i$ and one tracked susceptible $j$ infected by $i$, analogously to Figure~\ref{fig:pe_pi_pr}) and we study the probability that~$j$ is notified when in the \exposed, \infected, or \removed state, respectively, as we vary the testing delay. Note that, since we focus on a tagged pair of tracked people, this result does not depend on $\alpha$, which is a population-level parameter. As Figure~\ref{fig:pe_pi_pr} shows, as long as the test turnaround is less than 2 days, the infected person is most likely caught while they are still not contagious. Vice versa, beyond a 4-day turnaround, we basically intercept only people that are already contagious or even symptomatic. As discussed above, the earlier we intercept infected people, the better. Small testing delays are thus a key ingredient of a containment plan.

\begin{figure}[!t]
\begin{center}
\includegraphics[scale=0.6]{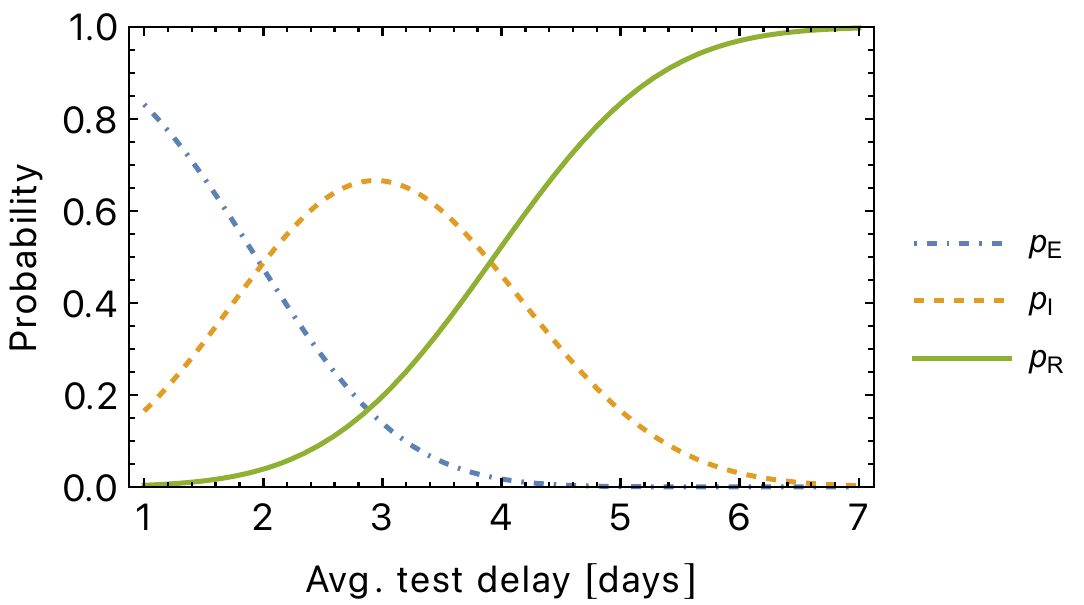}\vspace{-10pt}
\caption{$p_E$, $p_I$, and $p_R$ as the average testing delay increases.} \vspace{-20pt}
\label{fig:pe_pi_pr}
\end{center}
\end{figure}

\section{Can digital contact tracing tame an epidemic?}
\label{sec:evaluation}

Now we use the model defined in the previous section and apply it to study the effectiveness of digital contact tracing in controlling the COVID-19 epidemic. To this end, we need COVID-19-specific estimates for $\beta$ (effective contact rate), $\gamma$ (transition rate to symptomatic), and $\epsilon$ (transition rate to contagious), which are the baseline (i.e., not depending on the contact tracing process) transition rates in Figure~\ref{fig:SEIR_model}. Then, $\theta$ and $\psi$ (which instead are also determined by the contact tracing and the testing process) can be derived as discussed in Section~\ref{sec:parameters_from_contact_history}. 
Note that the digital contact tracing process itself can be described with two parameters: $\alpha$, the app uptake, and the test delay $\mu_T$. In the following, we will assess their impact on the control condition C1 (Theorem~\ref{theo:c1}) for different epidemic scenarios. Specifically, we consider three cases.  First, we test the effect of an increasing latent period. Then, we study the effect of a longer infectious period with constant transmissibility (see below for details). Finally, we consider what happens when transmissibility increases.

Regarding the baseline epidemic parameters, generally $\gamma$ and the ratio $\frac{\beta}{\gamma}$, corresponding to the basic reproduction number $R_0$, are estimated early in an epidemic. 
Hence, in the following, $\beta$ will be set to the value yielding the COVID-19 $R_0$ for the chosen $\gamma$. Intuitively, the basic reproduction number $R_0$ captures the transmissibility of a disease, i.e., the average number of cases directly generated by one infectious person in a population with a very large number of susceptibles. As we also see in the following, the larger $R_0$, the more difficult it is to contain the epidemic with digital contact tracing (and in general).


We start (Figure~\ref{fig:control_varying_epsilon}) with a scenario with an average contagious but asymptomatic window length ($\gamma^{-1}$) equal to 2 days (typical of COVID-19) and fixed $R_0 = 2$ (this value correspond to the initial 2020 estimate for COVID-19 in~\cite{Ferretti2020}). We test the effect of an increasing latent period length ($\epsilon^{-1} \in \{1, 3, 5\}$ days) on the controllability. Specifically, we plot condition C1 in Equation~\ref{eq:c1} as a function of app uptake $\alpha$ and testing delay (the two directions along which digital contact tracing can be improved): C1 holds true in the shadowed areas of the plot (therefore, the epidemic is under control). Intuitively, the longer the latent period, the easier is the control of the epidemic because we have more time to intercept tracked people before they become contagious. Figure~\ref{fig:control_varying_epsilon} confirms this: as $\epsilon^{-1}$ (the latent window) increases, the importance of small testing delays is reduced, and an app uptake above 80\% may be sufficient to control the outbreak. Looking again at Figure~\ref{fig:control_varying_epsilon}, we also note that an increase in $\epsilon^{-1}$ induces a temporal shift on the controllability boundary (the curves delimiting the shadowed areas in the figure): in other words, the controllability boundary is affected by a shift along the x-axis as the latent window increases.

Next, we consider the effect of the infectious period when the transmissibility is kept constant. Without contact tracing, two epidemics with the same transmissibility evolve in the same way (since the controllability condition reduces to $
R_0 < 1$). However, the length of the infectious period affects the effectiveness of contact tracing. Therefore, the digital controllability of epidemics with the same $R_0$ is different if their infectious periods are different. 
In Figure~\ref{fig:control_varying_gamma}, we set the latent window $\epsilon^{-1}$ at 3 days (its average value for COVID-19) and vary $\gamma$ (the duration of the infectious period) in $\{2, 5, 8\}$ days. Note that we want $R_0$ to remain fixed (no change in transmissibility), so we vary $\beta$ accordingly. This implies that each person, on average, infects the same number of people in all three cases captured by Figure~\ref{fig:control_varying_gamma}. 
We observe that the longer the contagious period, the easier the containment (since the shadowed area expands). 
This was expected because a longer contagious period increases the probability of catching infected people before they become symptomatic.
With respect to the controllability boundary (the curves delimiting the shadowed areas in the figure), a change in $\gamma$ induces a change in the convexity of the boundary but no horizontal shift. 

Finally, in Figure~\ref{fig:control_varying_beta} we fix the latent period $\epsilon^{-1}$ and the contagious period $\gamma^{-1}$ to their typical COVID-19 values of 3 and 2 days, respectively, and we change the $R_0$ of the epidemic by varying $\beta$. Note that this analysis is especially important, given the rise of new variants with increased transmissibility (and therefore greater $R_0$). We study $R_0 \in \{ 2, 3, 4, 5, 6\}$. $R_0=2$ is the initial estimate for COVID-19 (original strain), then revised to be much higher in some areas (e.g., the estimate in~\cite{Casella2020} is $R_0 \sim 4$). The Alpha variant (B.1.1.7 lineage) is estimated to feature at least 40\% higher R~\cite{Davies2021} with respect to the original strain, while the Delta variant (B.1.617.2) has transmissibility estimated between 6 and 7~\cite{Pung2021,Dagpunar2021}. Note that the apparent even higher transmissibility of the recent Omicron variant (B.1.1.529) seems to be due to immune evasion (e.g., vaccines not as effective as for previous variants) rather than to an actual increase in basic transmissibility~\cite{Lyngse2021.12.27.21268278}.
Figure~\ref{fig:control_varying_beta} shows that, as expected, the impact of increasing $R_0$ is much more disruptive than that of different latent/contagious windows. Specifically, even with instantaneous testing, the minimum uptake $\alpha$ needed to control the epidemic increases as $R_0$ increases. This means that even a small fraction of untracked people can wreak havoc on the containment measures.
In practice, however, with $R_0 = 4$, the control of the epidemic through digital contact tracing becomes impossible: an uptake above 95\% is unrealistic for all the reasons discussed in Section~\ref{sec:intro} (e.g., technical problems with old smartphones, distrust by a fraction of the population). In this case, $R_0$ must also be reduced by exploiting mitigation measures (social distancing, masks), in order to reduce the probability of infection upon contact, hence $\beta$.

\begin{figure}[!t]
\begin{center}
\includegraphics[scale=0.7]{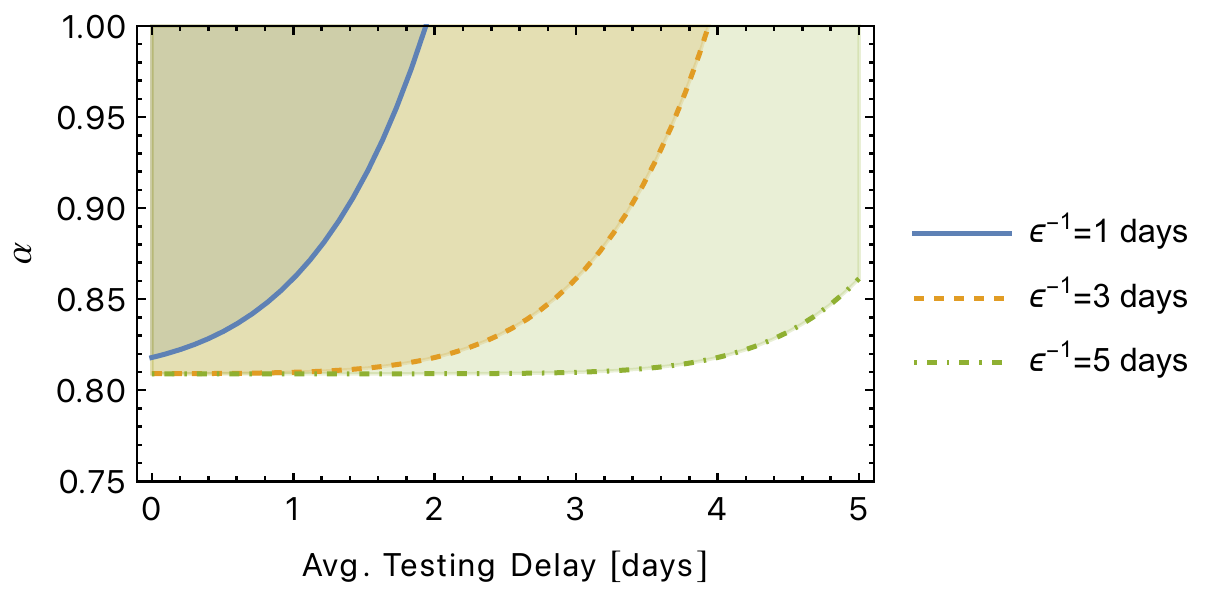}
\caption{Controllability when the average latent period increases (corresponding to the decrease of $\epsilon$). We consider $\epsilon^{-1} = \{1, 3, 5\}$ days. Control is attained in the shadowed regions.} \vspace{-10pt}
\label{fig:control_varying_epsilon}
\end{center}
\end{figure}

\begin{figure}[!t]
\begin{center}
\includegraphics[scale=0.7]{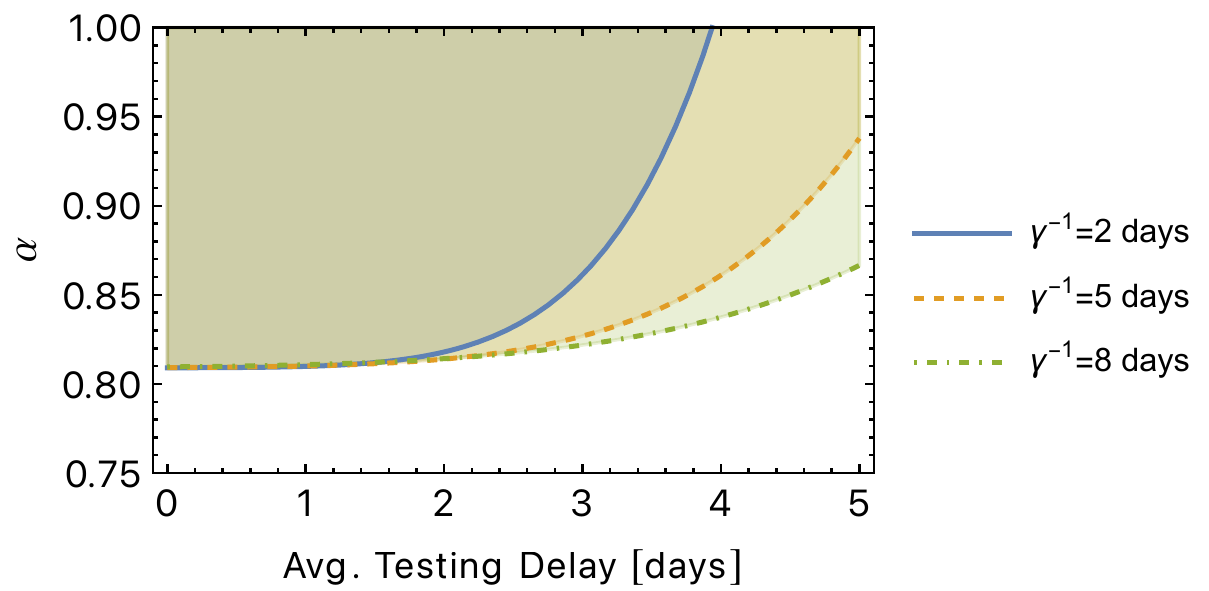}
\caption{Controllability when the average length of the contagious period increases (corresponding to $\gamma$ decreasing). We consider $\gamma^{-1} = \{2, 5, 8\}$ days. Control is attained in the shadowed regions.}\vspace{-10pt}
\label{fig:control_varying_gamma}
\end{center}
\end{figure}

\begin{figure}[!t]
\begin{center}
\includegraphics[scale=0.7]{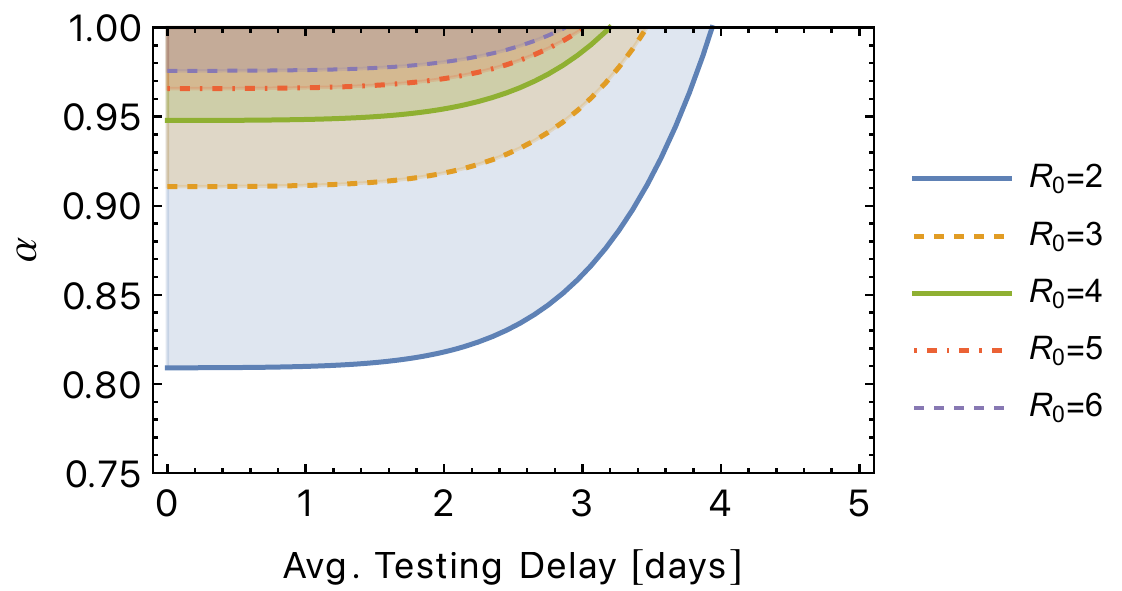}
\caption{Controllability when the $R_0 = \frac{\beta}{\gamma}$ increases (we fix $\gamma$ and we increase $\beta$). We study $R_0 \in \{ 2, 3, 4, 5, 6\}$. Control is attained in the shadowed regions.}\vspace{-20pt}
\label{fig:control_varying_beta}
\end{center}
\end{figure}

\section{Conclusion}
\label{sec:conclusion}

In this work, we have discussed the modeling efforts for COVID-19 and we have proposed a SEIR model that factors in digital contact tracing and is capable of producing a closed-form condition on the controllability of the epidemic. Leveraging this model, we have studied how the penetration of digital contact tracing apps within the population impacts the control of the epidemic. We have found that the penetration must be in general high, hence digital contact tracing may not be sufficient to contain an epidemic, even with a fast turnaround of tests. Additional mitigation strategies must be implemented, such as social distancing and mask wearing. Additionally, the impact of digital contact tracing is highest when the testing delay is low. If the test turnaround is greater than 4 days, digital contact tracing has zero impact on containment. In future work, we plan to extend the model to account for more variability among nodes, e.g., in terms of contact patterns, (along the lines of~\cite{Newman2002}), while still striving for closed-form controllability conditions.

\appendices

\section{Proof of Theorem~1}
\label{app:proof_c1}

\begin{proof}
We start by writing the ODE system corresponding to Figure~\ref{fig:SEIR_model}:
\begin{align}\label{eq:seir_system}
\frac{dS_U}{dt} &= - \beta (1-\alpha)(I_U + I_T) \nonumber\\
\frac{dE_U}{dt} &= \beta (1-\alpha)(I_U + I_T)- \epsilon E_U \nonumber\\
\frac{dI_U}{dt} &=  \epsilon E_U - \gamma I_U  \nonumber\\
\frac{dS_T}{dt} &= - \beta \alpha (I_U + I_T) \nonumber\\
\frac{dE_T}{dt} &= \beta \alpha I_U + (\beta \alpha - \theta) I_T - \epsilon E_T \nonumber\\
\frac{dI_T}{dt} &=  \epsilon E_T - (\gamma + \psi) I_T  \nonumber\\
\frac{dR}{dt} &= \gamma I_U + (\gamma+\theta + \psi) I_T .
\end{align}
The corresponding system of ODE can be rewritten in matrix form as $\mathbf{y'} = \mathbf{A} \mathbf{y}$, where $\mathbf{y} = [E_U, I_U, E_T, I_T]^T$ and $\mathbf{A}$ is given by the following:
\begin{align}
\mathbf{A} = \begin{bmatrix}
-\epsilon & \beta (1-\alpha)  & 0      & \beta (1-\alpha)\\
\epsilon & -\gamma & 0 & 0\\
0 & \beta \alpha & -\epsilon & \beta \alpha -\theta \\
0 & 0 & \epsilon & -\gamma - \psi
\end{bmatrix}.
\end{align}
The system in~\eqref{eq:seir_system} describes a dynamic system. Its stability (corresponding to the epidemic being under control or not) is assessed by studying its eigenvalues (see, e.g., \cite{Casella2020}), which correspond to the roots of the characteristic polynomial $p_A(x)$ of matrix A. 
In fact, since the solutions to a system of linear ODE $\mathbf{y'} = \mathbf{A} \mathbf{y}$ are of the form $y(t) = \sum_i c_i *e^{r_i t}$~\cite{braun1993differential} (where $c_i$'s are constants and $r_i$'s the eigenvalues/roots), it is clear that a positive root introduces instability into the system, because there would be an exponential function with a positive argument, hence an exponential growth in the epidemic. 
Therefore, we can study the roots of the characteristic polynomial $p_A(x)$ to assess under which conditions only negative roots exist. 
In order to avoid a trivial case, we assume $\beta > \gamma$ (i.e., the epidemic is not under control without contact tracing).
Using Descartes' rule of signs, we can derive the number of positive and negative roots of $p_A(x)$ without actually having to solve the polynomial (finding a closed form for the roots would not be feasible in this case). Starting with positive roots, we observe the following signs:
\begin{equation}\label{eq:coef_list}
\{ +, +, \sgn(k_2), \sgn(k_1), \sgn(k_0)\},
\end{equation}
where we have expressed $p_A(x)$ as $\sum_i k_i x^i$, $\sgn$ is the sign function (where $\sgn(\cdot)=1$ corresponds to sign $+$, $\sgn(\cdot)=-1$ corresponds to $-$), and $k_2, k_1, k_0$ are given by the following:
%
\begin{align}
k_2 &=& -\beta \epsilon + \epsilon^2 + \gamma (\gamma + \psi) + \epsilon (4 \gamma + 2 \psi + \theta), \nonumber{}\\
k_1 &=&\epsilon (2 \gamma + \psi + \theta) + \qquad \qquad \qquad \qquad \qquad \nonumber{} \\
& &  + \gamma [2 (\gamma + \psi) + \theta) -  \beta (\epsilon + \gamma + \psi - \psi \alpha ], \nonumber{}\\
k_0 &=& - [ (\beta - \gamma) (\gamma + \psi + \theta) ] + \beta (\psi + \theta) \alpha .
\label{eq:k_sign}
\end{align}
%
By studying the functions $\sgn(k_2), \sgn(k_1), \sgn(k_0)$, we obtain the following relationships between the coefficients' signs:
 \begin{equation}
 \sgn(k_0) = 1 \Rightarrow \sgn(k_1) = 1 \Rightarrow \sgn(k_2) = 1.
 \label{eq:conditions_chain}
 \end{equation}
In other words, since the rates must be all positive and $\alpha \in [0,1]$, when coefficient $k_0$ is positive, $k_2$ and $k_1$ must also be positive. This implies that not all possible sign permutations in Equation~\ref{eq:coef_list} are attainable, as illustrated in Table~\ref{tab:signs}. Discarding unattainable permutations, we can have at most one sign change across the coefficients of the polynomial, which implies at most one positive root. It thus follows that the condition under which we observe no sign changes is also the condition under which the epidemic can be controlled (zero positive roots). Thanks to Equation~\ref{eq:conditions_chain}, we know that $\sgn(k_0) = 1$ is a sufficient condition for this to happen. Then, solving for $k_0 > 0$ (with $k_0$ defined Equation~\ref{eq:k_sign}), we obtain condition C1 in Equation~\ref{eq:c1}.

To conclude the proof, we only need to verify that there are no complex roots. This is easy to do by applying again Descartes' rule, this time to $p_A(-x)$. To this aim, we need to change the coefficient sign of odd-power terms (i.e., $k3$, $k_1$) in Table~\ref{tab:signs} and count the sign changes. By summing the sign changes for positive and negative roots corresponding to the same permutation (equivalently, by summing the sign changes per corresponding row in Table~\ref{tab:signs} and Table~\ref{tab:signs} with the sign of odd-power terms changed), we obtain the total number of real roots. If we do the math, we discover that the total sign changes are at most $4$, hence  $p_A(x)$ features four real roots. Then, the number of complex roots is given by the difference between the degree of the polynomial and the maximum number of real roots. Since $p_A(x)$ is a polynomial of degree $4$, we know that there are no complex roots.
\begin{table}[t]
\caption{All possible sign permutations for the coefficients of the characteristic polynomial of $A$. The shadowed areas correspond to unfeasible permutations.} 
\label{tab:signs}
\centering
\begin{tabular}{*{6}{>{\raggedright\arraybackslash}p{0.11\columnwidth}}}
\toprule
$k_4$ & $k_3$ & $k_2$ & $k_1$ & $k_0$ & sign changes  \\ 
\cmidrule(lr){1-2}  \cmidrule(lr){3-5}  \cmidrule(lr){6-6}
+ & +  & + & + &  + & 0 \\
+ & +  & + & + & - &  1\\
\rowcolor[HTML]{DAE8FC} 
+ & +  & + &  - &  + &  2\\
+ & +  & + &  - &  - &  1\\
\rowcolor[HTML]{DAE8FC} 
+ & +  &  - &  + &  + &  2\\
\rowcolor[HTML]{DAE8FC} 
+ & +  &  - &  + &  - &  3\\
\rowcolor[HTML]{DAE8FC} 
+ & +  &  - &  - &  + &  2\\
+ & +  &  - &  - &  - & 1 \\ \bottomrule
\end{tabular}
\end{table}
%
\end{proof}
\ifCLASSOPTIONcaptionsoff
  \newpage
\fi



\bibliographystyle{IEEEtran}
\bibliography{covid19.bib}

\begin{thebibliography}{10}
\providecommand{\url}[1]{#1}
\csname url@samestyle\endcsname
\providecommand{\newblock}{\relax}
\providecommand{\bibinfo}[2]{#2}
\providecommand{\BIBentrySTDinterwordspacing}{\spaceskip=0pt\relax}
\providecommand{\BIBentryALTinterwordstretchfactor}{4}
\providecommand{\BIBentryALTinterwordspacing}{\spaceskip=\fontdimen2\font plus
\BIBentryALTinterwordstretchfactor\fontdimen3\font minus
  \fontdimen4\font\relax}
\providecommand{\BIBforeignlanguage}[2]{{%
\expandafter\ifx\csname l@#1\endcsname\relax
\typeout{** WARNING: IEEEtran.bst: No hyphenation pattern has been}%
\typeout{** loaded for the language `#1'. Using the pattern for}%
\typeout{** the default language instead.}%
\else
\language=\csname l@#1\endcsname
\fi
#2}}
\providecommand{\BIBdecl}{\relax}
\BIBdecl

\bibitem{Fraser2004}
C.~Fraser, S.~Riley, R.~M. Anderson, and N.~M. Ferguson, ``{Factors that make
  an infectious disease outbreak controllable},'' \emph{PNAS}, vol. 101,
  no.~16, pp. 6146--6151, 4 2004.

\bibitem{Cho2020}
\BIBentryALTinterwordspacing
H.~Cho, D.~Ippolito, and Y.~W. Yu, ``{Contact Tracing Mobile Apps for COVID-19:
  Privacy Considerations and Related Trade-offs},'' \emph{arXiv}, 2020.
  [Online]. Available: \url{http://arxiv.org/abs/2003.11511}
\BIBentrySTDinterwordspacing

\bibitem{Raskar2020}
\BIBentryALTinterwordspacing
R.~Raskar, I.~Schunemann, R.~Barbar, K.~Vilcans, J.~Gray, P.~Vepakomma,
  S.~Kapa, A.~Nuzzo, R.~Gupta, A.~Berke, D.~Greenwood, C.~Keegan, S.~Kanaparti,
  R.~Beaudry, D.~Stansbury, B.~B. Arcila, R.~Kanaparti, V.~Pamplona, F.~M.
  Benedetti, A.~Clough, R.~Das, K.~Jain, K.~Louisy, G.~Nadeau, V.~Pamplona,
  S.~Penrod, Y.~Rajaee, A.~Singh, G.~Storm, and J.~Werner, ``{Apps Gone Rogue:
  Maintaining Personal Privacy in an Epidemic},'' \emph{arXiv}, 2020. [Online].
  Available: \url{http://arxiv.org/abs/2003.08567}
\BIBentrySTDinterwordspacing

\bibitem{Nanni2020}
M.~Nanni, G.~Andrienko, A.-L. Barab{\'{a}}si, C.~Boldrini, F.~Bonchi,
  C.~Cattuto, F.~Chiaromonte, G.~Comand{\'{e}}, M.~Conti, M.~Cot{\'{e}},
  F.~Dignum, V.~Dignum, J.~Domingo-Ferrer, P.~Ferragina, F.~Giannotti,
  R.~Guidotti, D.~Helbing, K.~Kaski, J.~Kertesz, S.~Lehmann, B.~Lepri,
  P.~Lukowicz, S.~Matwin, D.~M. Jim{\'{e}}nez, A.~Monreale, K.~Morik,
  N.~Oliver, A.~Passarella, A.~Passerini, D.~Pedreschi, A.~Pentland,
  F.~Pianesi, F.~Pratesi, S.~Rinzivillo, S.~Ruggieri, A.~Siebes, R.~Trasarti,
  J.~van~den Hoven, and A.~Vespignani, ``{Give more data, awareness and control
  to individual citizens, and they will help COVID-19 containment},''
  \emph{Transactions on Data Privacy}, vol.~13, no.~1, pp. 61--66, 2020.

\bibitem{Oliver2020}
\BIBentryALTinterwordspacing
N.~Oliver, E.~Letouz{\'{e}}, H.~Sterly, S.~Delataille, M.~De~Nadai, B.~Lepri,
  R.~Lambiotte, R.~Benjamins, C.~Cattuto, V.~Colizza, N.~de~Cordes, S.~P.
  Fraiberger, T.~Koebe, S.~Lehmann, J.~Murillo, A.~Pentland, P.~N. Pham,
  F.~Pivetta, A.~A. Salah, J.~Saram{\"{a}}ki, S.~V. Scarpino, M.~Tizzoni,
  S.~Verhulst, and P.~Vinck, ``{Mobile phone data and COVID-19: Missing an
  opportunity?}'' \emph{arXiv}, 2020. [Online]. Available:
  \url{http://arxiv.org/abs/2003.12347}
\BIBentrySTDinterwordspacing

\bibitem{Ahmed2020}
N.~Ahmed, R.~A. Michelin, W.~Xue, S.~Ruj, R.~Malaney, S.~S. Kanhere,
  A.~Seneviratne, W.~Hu, H.~Janicke, and S.~K. Jha, ``{A Survey of COVID-19
  Contact Tracing Apps},'' pp. 134\,577--134\,601, 2020.

\bibitem{Troncoso2022}
\BIBentryALTinterwordspacing
C.~Troncoso, D.~Bogdanov, E.~Bugnion, S.~Chatel, C.~Cremers, S.~G\"{u}rses,
  J.-P. Hubaux, D.~Jackson, J.~R. Larus, W.~Lueks, R.~Oliveira, M.~Payer,
  B.~Preneel, A.~Pyrgelis, M.~Salath\'{e}, T.~Stadler, and M.~Veale,
  ``Deploying decentralized, privacy-preserving proximity tracing,''
  \emph{Commun. ACM}, vol.~65, no.~9, p. 48–57, aug 2022. [Online].
  Available: \url{https://doi.org/10.1145/3524107}
\BIBentrySTDinterwordspacing

\bibitem{Garg2020}
L.~Garg, E.~Chukwu, N.~Nasser, C.~Chakraborty, and G.~Garg, ``{Anonymity
  Preserving IoT-Based COVID-19 and Other Infectious Disease Contact Tracing
  Model},'' \emph{IEEE Access}, vol.~8, pp. 159\,402--159\,414, 2020.

\bibitem{Ferretti2020}
L.~Ferretti, C.~Wymant, M.~Kendall, L.~Zhao, A.~Nurtay, L.~Abeler-D{\"{o}}rner,
  M.~Parker, D.~Bonsall, and C.~Fraser, ``{Quantifying SARS-CoV-2 transmission
  suggests epidemic control with digital contact tracing},'' \emph{Science},
  vol. 368, p. eabb6936, 2020.

\bibitem{Li2020c}
\BIBentryALTinterwordspacing
T.~Li, C.~Cobb, Jackie, Yang, S.~Baviskar, Y.~Agarwal, B.~Li, L.~Bauer, and
  J.~I. Hong, ``{What Makes People Install a COVID-19 Contact-Tracing App?
  Understanding the Influence of App Design and Individual Difference on
  Contact-Tracing App Adoption Intention},'' \emph{Pervasive and Mobile
  Computing}, vol.~75, p. 101439, dec 2020. [Online]. Available:
  \url{http://arxiv.org/abs/2012.12415}
\BIBentrySTDinterwordspacing

\bibitem{Colizza2021}
\BIBentryALTinterwordspacing
V.~Colizza, E.~Grill, R.~Mikolajczyk, C.~Cattuto, A.~Kucharski, S.~Riley,
  M.~Kendall, K.~Lythgoe, D.~Bonsall, C.~Wymant, L.~Abeler-D{\"{o}}rner,
  L.~Ferretti, and C.~Fraser, ``{Time to evaluate COVID-19 contact-tracing
  apps},'' pp. 361--362, mar 2021. [Online]. Available:
  \url{https://doi.org/10.1038/s41591-021-01237-5}
\BIBentrySTDinterwordspacing

\bibitem{CONTI20122}
\BIBentryALTinterwordspacing
M.~Conti, S.~K. Das, C.~Bisdikian, M.~Kumar, L.~M. Ni, A.~Passarella,
  G.~Roussos, G.~Tr{\"o}ster, G.~Tsudik, and F.~Zambonelli, ``Looking ahead in
  pervasive computing: Challenges and opportunities in the era of
  cyber--physical convergence,'' \emph{Pervasive and Mobile Computing}, vol.~8,
  no.~1, pp. 2--21, 2012. [Online]. Available:
  \url{https://www.sciencedirect.com/science/article/pii/S1574119211001271}
\BIBentrySTDinterwordspacing

\bibitem{Bohmer2020}
M.~M. B{\"{o}}hmer, U.~Buchholz, V.~M. Corman, M.~Hoch, K.~Katz, D.~V.
  Marosevic, S.~B{\"{o}}hm, T.~Woudenberg, N.~Ackermann, R.~Konrad, U.~Eberle,
  B.~Treis, A.~Dangel, K.~Bengs, V.~Fingerle, A.~Berger,
  S.~H{\"{o}}rmansdorfer, S.~Ippisch, B.~Wicklein, A.~Grahl, K.~P{\"{o}}rtner,
  N.~Muller, N.~Zeitlmann, T.~S. Boender, W.~Cai, A.~Reich, M.~an~der Heiden,
  U.~Rexroth, O.~Hamouda, J.~Schneider, T.~Veith, B.~M{\"{u}}hlemann,
  R.~W{\"{o}}lfel, M.~Antwerpen, M.~Walter, U.~Protzer, B.~Liebl, W.~Haas,
  A.~Sing, C.~Drosten, and A.~Zapf, ``{Investigation of a COVID-19 outbreak in
  Germany resulting from a single travel-associated primary case: a case
  series},'' \emph{The Lancet Infectious Diseases}, 2020.

\bibitem{Hu2020}
Z.~Hu, C.~Song, C.~Xu, G.~Jin, Y.~Chen, X.~Xu, H.~Ma, W.~Chen, Y.~Lin,
  Y.~Zheng, J.~Wang, Z.~Hu, Y.~Yi, and H.~Shen, ``{Clinical characteristics of
  24 asymptomatic infections with COVID-19 screened among close contacts in
  Nanjing, China},'' \emph{Science China Life Sciences}, 2020.

\bibitem{Qian2020}
G.~Qian, N.~Yang, A.~H.~Y. Ma, L.~Wang, G.~Li, X.~Chen, and X.~Chen,
  ``{COVID-19 Transmission Within a Family Cluster by Presymptomatic Carriers
  in China},'' \emph{Clinical infectious diseases : an official publication of
  the Infectious Diseases Society of America}, 2020.

\bibitem{Rothe2020}
C.~Rothe, M.~Schunk, P.~Sothmann, G.~Bretzel, G.~Froeschl, C.~Wallrauch,
  T.~Zimmer, V.~Thiel, C.~Janke, W.~Guggemos, M.~Seilmaier, C.~Drosten,
  P.~Vollmar, K.~Zwirglmaier, S.~Zange, R.~W{\"{o}}lfel, and M.~Hoelscher,
  ``{Transmission of 2019-NCOV infection from an asymptomatic contact in
  Germany},'' 2020.

\bibitem{Wang2020}
Y.~Wang, J.~Tong, Y.~Qin, T.~Xie, J.~Li, J.~Li, J.~Xiang, Y.~Cui, E.~S. Higgs,
  J.~Xiang, and Y.~He, ``{Characterization of an asymptomatic cohort of
  SARS-COV-2 infected individuals outside of Wuhan, China},'' \emph{Clinical
  Infectious Diseases}, 2020.

\bibitem{Yu2020}
P.~Yu, J.~Zhu, Z.~Zhang, and Y.~Han, ``{A familial cluster of infection
  associated with the 2019 novel coronavirus indicating possible
  person-to-person transmission during the incubation period},'' \emph{Journal
  of Infectious Diseases}, 2020.

\bibitem{Bai2020}
Y.~Bai, L.~Yao, T.~Wei, F.~Tian, D.~Y. Jin, L.~Chen, and M.~Wang, ``{Presumed
  Asymptomatic Carrier Transmission of COVID-19},'' 2020.

\bibitem{Wei2020}
\BIBentryALTinterwordspacing
W.~E. Wei, Z.~Li, C.~J. Chiew, S.~E. Yong, M.~P. Toh, and V.~J. Lee,
  ``{Presymptomatic Transmission of SARS-CoV-2 --- Singapore, January 23--March
  16, 2020},'' \emph{MMWR. Morbidity and Mortality Weekly Report}, vol.~69,
  no.~14, 4 2020. [Online]. Available:
  \url{http://www.cdc.gov/mmwr/volumes/69/wr/mm6914e1.htm?s_cid=mm6914e1_w}
\BIBentrySTDinterwordspacing

\bibitem{Wolfel2020}
R.~W{\"{o}}lfel, V.~M. Corman, W.~Guggemos, M.~Seilmaier, S.~Zange, M.~A.
  M{\"{u}}ller, D.~Niemeyer, T.~C. Jones, P.~Vollmar, C.~Rothe, M.~Hoelscher,
  T.~Bleicker, S.~Br{\"{u}}nink, J.~Schneider, R.~Ehmann, K.~Zwirglmaier,
  C.~Drosten, and C.~Wendtner, ``{Virological assessment of hospitalized
  patients with COVID-2019},'' \emph{Nature}, pp. 1--10, 4 2020.

\bibitem{Kimball2020}
\BIBentryALTinterwordspacing
A.~Kimball, K.~M. Hatfield, M.~Arons, A.~James, J.~Taylor, K.~Spicer, A.~C.
  Bardossy, L.~P. Oakley, S.~Tanwar, Z.~Chisty, J.~M. Bell, M.~Methner,
  J.~Harney, J.~R. Jacobs, C.~M. Carlson, H.~P. McLaughlin, N.~Stone, S.~Clark,
  C.~Brostrom-Smith, L.~C. Page, M.~Kay, J.~Lewis, D.~Russell, B.~Hiatt,
  J.~Gant, J.~S. Duchin, T.~A. Clark, M.~A. Honein, S.~C. Reddy, J.~A.
  Jernigan, A.~Baer, L.~M. Barnard, E.~Benoliel, M.~S. Fagalde, J.~Ferro, H.~G.
  Smith, E.~Gonzales, N.~Hatley, G.~Hatt, M.~Hope, M.~Huntington-Frazier,
  V.~Kawakami, J.~L. Lenahan, M.~D. Lukoff, E.~B. Maier, S.~McKeirnan,
  P.~Montgomery, J.~L. Morgan, L.~A. Mummert, S.~Pogosjans, F.~X. Riedo,
  L.~Schwarcz, D.~Smith, S.~Stearns, K.~J. Sykes, H.~Whitney, H.~Ali, M.~Banks,
  A.~Balajee, E.~J. Chow, B.~Cooper, D.~W. Currie, J.~Dyal, J.~Healy,
  M.~Hughes, T.~M. McMichael, L.~Nolen, C.~Olson, A.~K. Rao, K.~Schmit, N.~G.
  Schwartz, F.~Tobolowsky, R.~Zacks, and S.~Zane, ``{Asymptomatic and
  Presymptomatic SARS-CoV-2 Infections in Residents of a Long-Term Care Skilled
  Nursing Facility --- King County, Washington, March 2020},'' \emph{MMWR.
  Morbidity and Mortality Weekly Report}, vol.~69, no.~13, pp. 377--381, 4
  2020. [Online]. Available:
  \url{http://www.cdc.gov/mmwr/volumes/69/wr/mm6913e1.htm?s_cid=mm6913e1_w}
\BIBentrySTDinterwordspacing

\bibitem{Li2020a}
R.~Li, S.~Pei, B.~Chen, Y.~Song, T.~Zhang, W.~Yang, and J.~Shaman,
  ``{Substantial undocumented infection facilitates the rapid dissemination of
  novel coronavirus (SARS-CoV2)},'' \emph{Science}, p. eabb3221, 3 2020.

\bibitem{Li2020b}
\BIBentryALTinterwordspacing
Q.~Li, X.~Guan, P.~Wu, X.~Wang, L.~Zhou, Y.~Tong, R.~Ren, K.~S. Leung, E.~H.
  Lau, J.~Y. Wong, X.~Xing, N.~Xiang, Y.~Wu, C.~Li, Q.~Chen, D.~Li, T.~Liu,
  J.~Zhao, M.~Liu, W.~Tu, C.~Chen, L.~Jin, R.~Yang, Q.~Wang, S.~Zhou, R.~Wang,
  H.~Liu, Y.~Luo, Y.~Liu, G.~Shao, H.~Li, Z.~Tao, Y.~Yang, Z.~Deng, B.~Liu,
  Z.~Ma, Y.~Zhang, G.~Shi, T.~T. Lam, J.~T. Wu, G.~F. Gao, B.~J. Cowling,
  B.~Yang, G.~M. Leung, and Z.~Feng, ``{Early Transmission Dynamics in Wuhan,
  China, of Novel Coronavirus--Infected Pneumonia},'' \emph{New England Journal
  of Medicine}, vol. 382, no.~13, pp. 1199--1207, 3 2020. [Online]. Available:
  \url{http://www.nejm.org/doi/10.1056/NEJMoa2001316}
\BIBentrySTDinterwordspacing

\bibitem{Sutton2020}
\BIBentryALTinterwordspacing
D.~Sutton, K.~Fuchs, M.~D'Alton, and D.~Goffman, ``{Universal Screening for
  SARS-CoV-2 in Women Admitted for Delivery},'' \emph{New England Journal of
  Medicine}, p. NEJMc2009316, 4 2020. [Online]. Available:
  \url{http://www.nejm.org/doi/10.1056/NEJMc2009316}
\BIBentrySTDinterwordspacing

\bibitem{He2020}
\BIBentryALTinterwordspacing
X.~He, E.~H. Lau, P.~Wu, X.~Deng, J.~Wang, X.~Hao, Y.~C. Lau, J.~Y. Wong,
  Y.~Guan, X.~Tan, X.~Mo, Y.~Chen, B.~Liao, W.~Chen, F.~Hu, Q.~Zhang, M.~Zhong,
  Y.~Wu, L.~Zhao, F.~Zhang, B.~J. Cowling, F.~Li, and G.~M. Leung, ``{Temporal
  dynamics in viral shedding and transmissibility of COVID-19},''
  \emph{medRxiv}, p. 2020.03.15.20036707, 3 2020. [Online]. Available:
  \url{https://www.medrxiv.org/content/10.1101/2020.03.15.20036707v2}
\BIBentrySTDinterwordspacing

\bibitem{Luo2020}
L.~Luo, D.~Liu, X.-l. Liao, X.-b. Wu, Q.-l. Jing, J.-z. Zheng, F.-h. Liu, S.-g.
  Yang, B.~Bi, Z.-h. Li, J.-p. Liu, W.-q. Song, W.~Zhu, Z.-h. Wang, X.-r.
  Zhang, P.-l. Chen, H.-m. Liu, X.~Cheng, M.-c. Cai, Q.-m. Huang, P.~Yang,
  X.-f. Yang, Z.-g. Huang, J.-l. Tang, Y.~Ma, and C.~Mao, ``{Modes of contact
  and risk of transmission in COVID-19 among close contacts},'' \emph{medRxiv},
  p. 2020.03.24.20042606, 3 2020.

\bibitem{Morawska2020}
L.~Morawska and D.~K. Milton, ``{It is Time to Address Airborne Transmission of
  COVID-19.}'' \emph{Clinical infectious diseases : an official publication of
  the Infectious Diseases Society of America}, 2020.

\bibitem{chagla2020airborne-transmission}
\BIBentryALTinterwordspacing
Z.~Chagla, S.~Hota, S.~Khan, and D.~Mertz, ``{Airborne Transmission of
  COVID-19},'' \emph{Clinical Infectious Diseases}, 2020. [Online]. Available:
  \url{https://pubmed.ncbi.nlm.nih.gov/32780799/}
\BIBentrySTDinterwordspacing

\bibitem{Klompas2020}
M.~Klompas, M.~A. Baker, and C.~Rhee, ``{Airborne Transmission of SARS-CoV-2:
  Theoretical Considerations and Available Evidence},'' 2020.

\bibitem{Bazant2021}
\BIBentryALTinterwordspacing
M.~Z. Bazant and J.~W. Bush, ``{A guideline to limit indoor airborne
  transmission of COVID-19},'' \emph{Proceedings of the National Academy of
  Sciences of the United States of America}, vol. 118, no.~17, apr 2021.
  [Online]. Available: \url{https://doi.org/10.1073/pnas.2018995118}
\BIBentrySTDinterwordspacing

\bibitem{Bourouiba2020}
L.~Bourouiba, ``{Turbulent Gas Clouds and Respiratory Pathogen Emissions:
  Potential Implications for Reducing Transmission of COVID-19},'' pp. E1--E2,
  2020.

\bibitem{Hunter2018}
\BIBentryALTinterwordspacing
E.~Hunter, B.~M. Namee, and J.~Kelleher, ``{An open-data-driven agent-based
  model to simulate infectious disease outbreaks},'' \emph{PLoS ONE}, vol.~13,
  no.~12, p. e0208775, 12 2018. [Online]. Available:
  \url{http://dx.plos.org/10.1371/journal.pone.0208775}
\BIBentrySTDinterwordspacing

\bibitem{Venkatramanan2018}
S.~Venkatramanan, B.~Lewis, J.~Chen, D.~Higdon, A.~Vullikanti, and M.~Marathe,
  ``{Using data-driven agent-based models for forecasting emerging infectious
  diseases},'' \emph{Epidemics}, vol.~22, pp. 43--49, 3 2018.

\bibitem{Hackl2019}
J.~Hackl and T.~Dubernet, ``{Epidemic spreading in urban areas using
  agent-based transportation models},'' \emph{Future Internet}, vol.~11, no.~4,
  p.~92, 4 2019.

\bibitem{Parker2011}
J.~Parker and J.~M. Epstein, ``{A distributed platform for global-scale
  agent-based models of disease transmission},'' \emph{ACM Transactions on
  Modeling and Computer Simulation}, vol.~22, no.~1, 12 2011.

\bibitem{Bonabeau2002}
E.~Bonabeau, ``{Agent-based modeling: Methods and techniques for simulating
  human systems},'' \emph{Proceedings of the National Academy of Sciences of
  the United States of America}, vol.~99, no. SUPPL. 3, pp. 7280--7287, 5 2002.

\bibitem{Grefenstette2013}
J.~J. Grefenstette, S.~T. Brown, R.~Rosenfeld, J.~Depasse, N.~T. Stone, P.~C.
  Cooley, W.~D. Wheaton, A.~Fyshe, D.~D. Galloway, A.~Sriram, H.~Guclu,
  T.~Abraham, and D.~S. Burke, ``{FRED (A Framework for Reconstructing Epidemic
  Dynamics): An open-source software system for modeling infectious diseases
  and control strategies using census-based populations},'' \emph{BMC Public
  Health}, vol.~13, no.~1, pp. 1--14, 10 2013.

\bibitem{Liu2015}
F.~Liu, W.~T. Enanoria, J.~Zipprich, S.~Blumberg, K.~Harriman, S.~F. Ackley,
  W.~D. Wheaton, J.~L. Allpress, and T.~C. Porco, ``{The role of vaccination
  coverage, individual behaviors, and the public health response in the control
  of measles epidemics: An agent-based simulation for California},'' \emph{BMC
  Public Health}, vol.~15, no.~1, pp. 1--16, 5 2015.

\bibitem{Tomy2022}
\BIBentryALTinterwordspacing
A.~Tomy, M.~Razzanelli, F.~{Di Lauro}, D.~Rus, and C.~{Della Santina},
  ``{Estimating the state of epidemics spreading with graph neural networks},''
  \emph{Nonlinear Dynamics}, vol. 109, no.~1, pp. 249--263, 2022. [Online].
  Available: \url{https://doi.org/10.1007/s11071-021-07160-1}
\BIBentrySTDinterwordspacing

\bibitem{Fritz2022}
\BIBentryALTinterwordspacing
C.~Fritz, E.~Dorigatti, and D.~R{\"{u}}gamer, ``{Combining graph neural
  networks and spatio-temporal disease models to improve the prediction of
  weekly COVID-19 cases in Germany},'' \emph{Scientific Reports}, vol.~12,
  no.~1, pp. 1--18, 2022. [Online]. Available:
  \url{https://doi.org/10.1038/s41598-022-07757-5}
\BIBentrySTDinterwordspacing

\bibitem{Gao2021}
J.~Gao, R.~Sharma, C.~Qian, L.~M. Glass, J.~Spaeder, J.~Romberg, J.~Sun, and
  C.~Xiao, ``{STAN: spatio-temporal attention network for pandemic prediction
  using real-world evidence},'' \emph{Journal of the American Medical
  Informatics Association : JAMIA}, vol.~28, no.~4, pp. 733--743, 2021.

\bibitem{Kapoor2020}
\BIBentryALTinterwordspacing
A.~Kapoor, X.~Ben, L.~Liu, B.~Perozzi, M.~Barnes, M.~Blais, and S.~O'Banion,
  ``{Examining COVID-19 Forecasting using Spatio-Temporal Graph Neural
  Networks},'' 2020. [Online]. Available: \url{http://arxiv.org/abs/2007.03113}
\BIBentrySTDinterwordspacing

\bibitem{Cutura2021}
G.~{\v{C}}utura, B.~Li, A.~Swami, and S.~Segarra, ``{Deep Demixing:
  Reconstructing the Evolution of Epidemics using Graph Neural Networks},''
  \emph{European Signal Processing Conference}, vol. 2021-August, pp.
  2204--2208, 2021.

\bibitem{Shah2020}
\BIBentryALTinterwordspacing
C.~Shah, N.~Dehmamy, N.~Perra, M.~Chinazzi, A.-L. Barab{\'{a}}si,
  A.~Vespignani, and R.~Yu, ``{Finding Patient Zero: Learning Contagion Source
  with Graph Neural Networks},'' pp. 1--18, 2020. [Online]. Available:
  \url{http://arxiv.org/abs/2006.11913}
\BIBentrySTDinterwordspacing

\bibitem{Dash2021}
S.~Dash, C.~Chakraborty, S.~K. Giri, S.~K. Pani, and J.~Frnda, ``{BIFM:
  Big-Data Driven Intelligent Forecasting Model for COVID-19},'' \emph{IEEE
  Access}, vol.~9, pp. 97\,505--97\,517, 2021.

\bibitem{Brauer2012}
F.~Brauer and C.~Castillo-Chavez, \emph{{Mathematical Models in Population
  Biology and Epidemiology}}.\hskip 1em plus 0.5em minus 0.4em\relax Springer,
  2012.

\bibitem{Giordano2020ModellingItaly.}
G.~Giordano, F.~Blanchini, R.~Bruno, P.~Colaneri, A.~Di~Filippo, A.~Di~Matteo,
  and M.~Colaneri, ``{Modelling the COVID-19 epidemic and implementation of
  population-wide interventions in Italy.}'' \emph{Nature medicine}, pp. 1--6,
  4 2020.

\bibitem{Newman2005}
M.~E. Newman, ``{Threshold effects for two pathogens spreading on a network},''
  \emph{Physical Review Letters}, vol.~95, no.~10, 2005.

\bibitem{Riou2020}
J.~Riou and C.~L. Althaus, ``{Pattern of early human-to-human transmission of
  Wuhan 2019 novel coronavirus (2019-nCoV), December 2019 to January 2020},''
  \emph{Euro surveillance}, vol.~25, no.~4, 1 2020.

\bibitem{ferguson2020report}
N.~Ferguson, D.~Laydon, G.~Nedjati~Gilani, N.~Imai, K.~Ainslie, M.~Baguelin,
  S.~Bhatia, A.~Boonyasiri, Z.~Cucunuba~Perez, G.~Cuomo-Dannenburg, and
  {others}, ``{Report 9: Impact of non-pharmaceutical interventions (NPIs) to
  reduce COVID19 mortality and healthcare demand},'' Imperial College London,
  Tech. Rep., 2020.

\bibitem{Linton2020EpidemiologicalData}
\BIBentryALTinterwordspacing
N.~M. Linton, T.~Kobayashi, Y.~Yang, K.~Hayashi, A.~R. Akhmetzhanov, S.-m.
  Jung, B.~Yuan, R.~Kinoshita, and H.~Nishiura, ``{Epidemiological
  characteristics of novel coronavirus infection: A statistical analysis of
  publicly available case data},'' \emph{medRxiv}, 2020. [Online]. Available:
  \url{https://doi.org/10.1101/2020.01.26.20018754}
\BIBentrySTDinterwordspacing

\bibitem{Li2020}
Q.~Li, X.~Guan, P.~Wu, X.~Wang, L.~Zhou, Y.~Tong, R.~Ren, K.~S. Leung, E.~H.
  Lau, J.~Y. Wong, X.~Xing, N.~Xiang, Y.~Wu, C.~Li, Q.~Chen, D.~Li, T.~Liu,
  J.~Zhao, M.~Liu, W.~Tu, C.~Chen, L.~Jin, R.~Yang, Q.~Wang, S.~Zhou, R.~Wang,
  H.~Liu, Y.~Luo, Y.~Liu, G.~Shao, H.~Li, Z.~Tao, Y.~Yang, Z.~Deng, B.~Liu,
  Z.~Ma, Y.~Zhang, G.~Shi, T.~T. Lam, J.~T. Wu, G.~F. Gao, B.~J. Cowling,
  B.~Yang, G.~M. Leung, and Z.~Feng, ``{Early Transmission Dynamics in Wuhan,
  China, of Novel Coronavirus--Infected Pneumonia},'' \emph{New England Journal
  of Medicine}, 1 2020.

\bibitem{Dorigatti}
\BIBentryALTinterwordspacing
I.~Dorigatti, L.~Okell, A.~Cori, N.~Imai, M.~Baguelin, S.~Bhatia,
  A.~Boonyasiri, Z.~Cucunub{\'{a}}, G.~Cuomo-Dannenburg, R.~Fitzjohn, H.~Fu,
  K.~Gaythorpe, A.~Hamlet, W.~Hinsley, N.~Hong, M.~Kwun, D.~Laydon,
  G.~Nedjati-Gilani, S.~Riley, S.~Van~Elsland, E.~Volz, H.~Wang, R.~Wang,
  C.~Walters, X.~Xi, C.~Donnelly, A.~Ghani, and N.~Ferguson, ``{Report 4:
  Severity of 2019-novel coronavirus (nCoV)},'' Imperial College London, Tech.
  Rep., 2020. [Online]. Available: \url{https://doi.org/10.25561/77154}
\BIBentrySTDinterwordspacing

\bibitem{Arenas2020}
A.~Arenas, W.~Cota, J.~Gomez-Gardenes, S.~Gomez, C.~Granell, J.~T. Matamalas,
  D.~Soriano-Panos, and B.~Steinegger, ``{Derivation of the effective
  reproduction number R for COVID-19 in relation to mobility restrictions and
  confinement},'' \emph{medRxiv}, p. 2020.04.06.20054320, 2020.

\bibitem{daley2001epidemic}
D.~J. Daley and J.~Gani, \emph{Epidemic modelling: an introduction}.\hskip 1em
  plus 0.5em minus 0.4em\relax Cambridge University Press, 1999.

\bibitem{dekking2005modern}
F.~M. Dekking, C.~Kraaikamp, H.~P. Lopuha{\"a}, and L.~E. Meester, \emph{A
  Modern Introduction to Probability and Statistics: Understanding why and
  how}.\hskip 1em plus 0.5em minus 0.4em\relax Springer, 2005.

\bibitem{boldrini2011pareto}
C.~Boldrini, M.~Conti, and A.~Passarella, ``From pareto inter-contact times to
  residuals,'' \emph{IEEE communications letters}, vol.~15, no.~11, pp.
  1256--1258, 2011.

\bibitem{Bar-On2020}
Y.~M. Bar-On, A.~Flamholz, R.~Phillips, and R.~Milo, ``{Sars-cov-2 (Covid-19)
  by the numbers},'' \emph{eLife}, 2020.

\bibitem{Casella2020}
F.~Casella, ``Can the covid-19 epidemic be controlled on the basis of daily
  test reports?'' \emph{IEEE Control Systems Letters}, vol.~5, no.~3, pp.
  1079--1084, 2020.

\bibitem{Davies2021}
\BIBentryALTinterwordspacing
N.~G. Davies, S.~Abbott, R.~C. Barnard, C.~I. Jarvis, A.~J. Kucharski, J.~D.
  Munday, C.~A.~B. Pearson, T.~W. Russell, D.~C. Tully, A.~D. Washburne,
  T.~Wenseleers, A.~Gimma, W.~Waites, K.~L.~M. Wong, K.~van Zandvoort, J.~D.
  Silverman, K.~Diaz-Ordaz, R.~Keogh, R.~M. Eggo, S.~Funk, M.~Jit, K.~E.
  Atkins, and W.~J. Edmunds, ``{Estimated transmissibility and impact of
  SARS-CoV-2 lineage B.1.1.7 in England},'' \emph{Science}, vol. 372, no. 6538,
  p. eabg3055, apr 2021. [Online]. Available:
  \url{https://doi.org/10.1126/science.abg3055}
\BIBentrySTDinterwordspacing

\bibitem{Pung2021}
\BIBentryALTinterwordspacing
R.~Pung, T.~{Minn Mak}, C.~Covid, working Group, A.~J. Kucharski, and V.~J.
  Lee, ``{Serial intervals observed in SARS-CoV-2 B.1.617.2 variant cases},''
  \emph{medRxiv}, p. 2021.06.04.21258205, jun 2021. [Online]. Available:
  \url{https://doi.org/10.1101/2021.06.04.21258205}
\BIBentrySTDinterwordspacing

\bibitem{Dagpunar2021}
\BIBentryALTinterwordspacing
J.~Dagpunar, ``{Interim estimates of increased transmissibility, growth rate,
  and reproduction number of the Covid-19 B.1.617.2 variant of concern in the
  United Kingdom},'' \emph{medRxiv}, p. 2021.06.03.21258293, jun 2021.
  [Online]. Available: \url{https://doi.org/10.1101/2021.06.03.21258293}
\BIBentrySTDinterwordspacing

\bibitem{Lyngse2021.12.27.21268278}
\BIBentryALTinterwordspacing
F.~P. Lyngse, L.~H. Mortensen, M.~J. Denwood, L.~E. Christiansen, C.~H.
  M{\o}ller, R.~L. Skov, K.~Spiess, A.~Fomsgaard, M.~M. Lassauni{\`e}re,
  M.~Rasmussen, M.~Stegger, C.~Nielsen, R.~N. Sieber, A.~S. Cohen, F.~T.
  M{\o}ller, M.~Overvad, K.~M{\o}lbak, T.~G. Krause, and C.~T. Kirkeby,
  ``Sars-cov-2 omicron voc transmission in danish households,'' \emph{medRxiv},
  2021. [Online]. Available:
  \url{https://www.medrxiv.org/content/early/2021/12/27/2021.12.27.21268278}
\BIBentrySTDinterwordspacing

\bibitem{Newman2002}
M.~E. Newman, ``{Spread of epidemic disease on networks},'' \emph{Physical
  Review E - Statistical Physics, Plasmas, Fluids, and Related
  Interdisciplinary Topics}, vol.~66, no.~1, pp. 1--11, 2002.

\bibitem{braun1993differential}
M.~Braun, \emph{Differential Equations and Their Applications An Introduction
  to Applied Mathematics}.\hskip 1em plus 0.5em minus 0.4em\relax Springer,
  1993.

\end{thebibliography}
%






\begin{IEEEbiography}[{\includegraphics[width=1in,height=1.25in,clip,keepaspectratio]{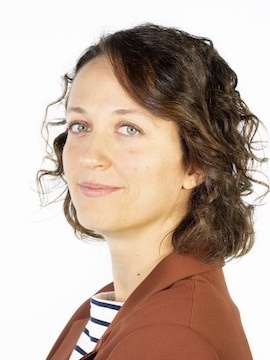}}]{Chiara Boldrini}
is a Researcher at IIT-CNR. Her research interests are in computational social sciences, decentralized AI, urban data science, mobile and ubiquitous systems. She has published 50+ papers on these topics. She is the IIT-CNR co-PI for H2020 SoBigData++ and HumaneE-AI-Net projects, and was involved in several EC projects since FP7. She is in the Editorial Board of Elsevier Pervasive and Mobile Computing and of Elsevier Computer Communications. She was the lead guest editor for the PMC Special Issue on IoT for Fighting COVID-19. She has served as TPC chair of IEEE SmartComp'22, as TPC vice-chair of IEEE PerCom'21 and, over the years, has been in the organizing committee of  several IEEE and ACM conferences/workshops, including IEEE PerCom and ACM MobiHoc.
\end{IEEEbiography}

\vfill

\begin{IEEEbiography}[{\includegraphics[width=1in,height=1.25in,clip,keepaspectratio]{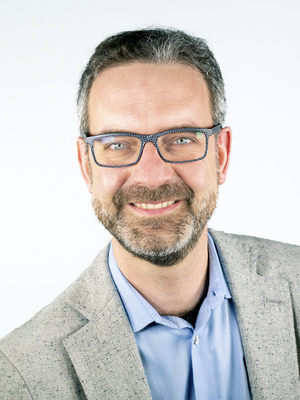}}]{Andrea Passarella} 
Andrea Passarella (Ph.D 2005) is a Research Director at the Institute for Informatics and Telematics (IIT) of the National Research Council of Italy (CNR). Prior to join IIT he was with the Computer Laboratory of the University of Cambridge, UK. He has published 200+ papers on Online and Mobile social networks, decentralised AI, Next Generation Internet, opportunistic, ad hoc and sensor networks, receiving the best paper award at IFIP Networking 2011 and IEEE WoWMoM 2013. He served as General Chair for IEEE PerCom 2022. He is the founding Associate Editor-in-Chief of Elsevier Online Social Networks. He is co-author of the book "Online Social Networks: Human Cognitive Constraints in Facebook and Twitter Personal Graphs" (Elsevier, 2015), and was Guest Co-Editor of several special sections in ACM and Elsevier Journals.  He is the PI of the EU CHIST-ERA SAI (Social Explainable AI) project.
\end{IEEEbiography}

\begin{IEEEbiography}[{\includegraphics[width=1in,height=1.25in,clip,keepaspectratio]{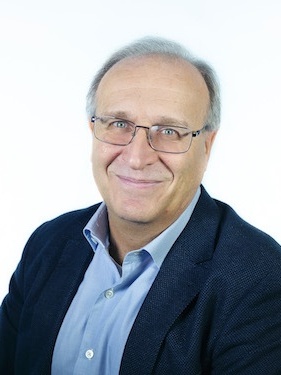}}]{Marco Conti}
is a CNR research director and, currently, he is the director of IIT-CNR institute. He has published more than 400 scientific articles related to design, modeling, and experimentation of computer and communication networks, pervasive systems, and online social networks. He is the founding EiC of Online Social Networks and Media, EiC for Special Issues of Pervasive and Mobile Computing, and, from 2009 to 2018, EiC of Computer Communications. He has received several awards, including the Best Paper Award at IFIP Networking 2011, IEEE ISCC 2012, and IEEE WoWMoM 2013. He was included in the "2017 Highly Cited Researchers" list compiled by Web of Science for the most cited articles in Computer Science. He served as the General/Program chair of several major conferences, including IFIP Networking 2002, IEEE WoWMoM 2005 and 2006, IEEE PerCom 2006 and 2010, ACM MobiHoc 2006, IEEE MASS 2007 and IEEE SmartComp 2021.
\end{IEEEbiography}

\vfill





 \enlargethispage{-2in}


\end{document}